%% file: Gradient_Quantization.tex
\documentclass[12pt, draftclsnofoot, onecolumn]{IEEEtran}	

 \usepackage{amsmath,amssymb}
 \usepackage{subfigure}
 \usepackage{graphicx,graphics,color,psfrag}
 \usepackage{cite,balance}
 \usepackage{algorithm}
 \usepackage{accents}
 \usepackage{amsthm}
 \usepackage{bm}
 \usepackage{url}
 \usepackage{algorithmic}
 \usepackage[english]{babel}
 \usepackage{multirow}
 \usepackage{enumerate}
 \usepackage{cases}
 \usepackage{stfloats}
 \usepackage{dsfont}
 \usepackage{color,soul}
 \usepackage{amsfonts}
 \usepackage{cite,graphicx,amsmath,amssymb}
 \usepackage{subfigure}
 \usepackage{fancyhdr}
 \usepackage{hhline}
 \usepackage{graphicx,graphics}
 \usepackage{array,color}
 \usepackage{amsmath}
 \usepackage{amsthm}
  \usepackage{booktabs}
\usepackage{framed}

\include{header}

\newcommand{\Quan}{\mathsf{Q}}

\begin{document}

\title{\huge High-Dimensional Stochastic Gradient Quantization for Communication-Efficient Edge Learning}

\author{Yuqing Du, Sheng Yang and Kaibin Huang \vspace{-2mm} \\
\thanks{ Y. Du and K.~Huang are with  The University of Hong Kong, Hong
Kong (Email: yqdu@eee.hku.hk, huangkb@eee.hku.hk). S.Yang is with
Laboratory of Signals and Systems, CentraleSup\'elec, University of Paris-Saclay, 91190 Gif-sur-Yvette, France (e-mail: sheng.yang@centralesupelec.fr).}
}

\maketitle

\vspace{-20pt}

\begin{abstract}
Edge machine learning involves the deployment of learning algorithms at the wireless network edge so as to leverage massive mobile data  for enabling intelligent applications. The mainstream edge learning approach, federated learning, has been developed based on distributed gradient descent. Based on the approach, stochastic gradients are computed at edge devices and then transmitted to an edge server for updating a global AI model. Since each stochastic gradient is typically high-dimensional (with millions to billions of coefficients), communication overhead becomes a bottleneck for edge learning. To address this issue, we propose in this work a novel framework of hierarchical stochastic gradient quantization and study its effect on the learning performance.  First, the framework features a practical hierarchical architecture for decomposing the stochastic gradient into its norm and normalized block gradients, and efficiently quantizes them using a uniform quantizer and a low-dimensional codebook on a Grassmann manifold, respectively. Subsequently, the quantized normalized block gradients are scaled and cascaded to yield the quantized normalized stochastic gradient using a so-called hinge vector designed under the criterion of minimum distortion. The hinge vector is also efficiently compressed using another low-dimensional Grassmannian quantizer. The other feature of the framework is a bit-allocation scheme for reducing the quantization error. The scheme divides the total bits from gradient quantization to determine the resolutions of the low-dimensional quantizers in the proposed framework. The framework is proved to guarantee model convergency by analyzing the convergence rate as a function of the quantization bits. Furthermore, by simulation, our design is shown to substantially reduce the communication overhead compared with the state-of-the-art signSGD scheme, while both achieve similar learning accuracies.

\end{abstract}

\section{introduction}
Recently, a large amount of data are generated in real-time and distributed at edge devices (e.g., smartphones and sensors). To allow rapid access to the enormous real-time data generated by edge devices for \emph{artificial intelligence} (AI)-model training, several edge learning frameworks such as \emph{federated edge learning} (FEEL) have been developed based on distributed stochastic gradient descent~\cite{zhu2018towards,konevcny2016federated}. Based on the approach, stochastic gradients are computed at edge devices and then transmitted to an edge server for aggregation and then updating a global AI-model. Typical stochastic gradients are of high dimensionality (each constitutes e.g., millions of parameters). Thus their transmission over communication networks can result in extensive overhead and a bottleneck for fast edge learning. To tackle this challenge, numerous schemes have been developed for  compressing stochastic gradients to reduce the said communication overhead~\cite{tandon2017gradient,chen2018lag,zhang2017zipml}. However, due to gradients' high dimensionality, most of the existing schemes focus merely on scalar quantization. The area of high-dimensional \emph{vector quantization} (VQ) targeting stochastic gradients is largely uncharted. This motivates us to make the first attempt on filling the void. Specifically, we propose a novel hierarchical vector quantization scheme using low-dimensional Grassmannian codebooks. By both simulation and theoretic analyses, this scheme is shown to be of low-complexity, communication-efficient, and guarantee learning convergence.

\subsection{Stochastic Gradient Quantization}
Recently, the topic of stochastic gradient quantization has been attracting growing interests for its being a key approach for improving the communication efficiency of edge learning~\cite{alistarh2017qsgd,2018signsgd,wu2018error,zheng2019communication, kingma2014adam, gersho2012vector}. In~\cite{alistarh2017qsgd}, a scheme called ``Quantized SGD" (QSGD) is proposed, where a scalar quantizer is deployed and its efficiency is improved by Elias integer coding exploiting the distribution of quantized gradient values. Building on QSGD, the effect of the quantization error can be further alleviated using an error-compensation scheme presented in~\cite{wu2018error}. To be specific, the accumulated quantization error is exploited to accelerate the model convergence. A recent key advancement in the area is the finding that despite its supposed low resolution, the combination of one-bit scalar quantizer for gradient-coefficient quantization, named ``signSGD"~\cite{2018signsgd}, and momentum in descent can be proved to achieve a convergence rate of the same order as its counterpart without quantization, namely the famous ``ADAM" scheme~\cite{kingma2014adam}. The promising result has motivated a series of followup work. For example, the original signSGD can be improved by dividing the large-number of one-bit coefficients of a quantized gradient into blocks and scaling each block by the norm of the unquantized counterpart~\cite{zheng2019communication}. Such modifications are shown to accelerate learning. All the above  schemes are based on scalar quantization. There are few results on the VQ of gradients despite its being a well-developed area~\cite{gersho2012vector}.

VQ, namely the joint quantization of the entries of a vector, is required to achieve the optimal rate-distortion trade-off~\cite{gersho2012vector}. Such asymptotic optimality, however, comes at the price of an exponentially growing complexity with the vector length. This makes it infeasible to directly apply the classic VQ algorithms to the quantization of high-dimensional stochastic gradients and explains the current popularity of scalar quantization for stochastic gradient quantization in edge learning. However, the effectiveness of VQ proven in conventional data compression suggests its potential for improving the communication efficiency for edge learning. This motivates the current work on designing a new VQ framework for stochastic gradient compression targeting SGD.

\subsection{Grassmannian Quantization in Wireless Communication}
A Grassmann manifold refers to a space of lines or subspaces embedded in a higher-dimensional space. A quantizer for partitioning the manifold typically uses a Grassmannian codebook that comprises a set of lines or subspaces and a subspace distance as the distortion measure (e.g., the sine of two lines' separation angle). Consequently, the quantizer attempts to minimize the deviation in direction between a line and its quantized version, or the deviation in orientation for the case of a subspace. Thus, Grassmannian quantization is a suitable tool for compressing data containing information on vector direction or subspace orientation. 

In wireless communication, Grassmannian quantization is widely adopted in one particular area, limited feedback, for efficient feedback of a quantized beamformer/precoder from a receiver to a transmitter to enable adaptive multi-antenna transmission~\cite{love2008overview,mukkavilli2003beamforming,love2003grassmannian,raghavan2007systematic,kim2011mimo,huang2009limited}. In~\cite{mukkavilli2003beamforming}, for beamformer quantization and feedback, a randomly generated Grassmannian codebook, which comprises a set of unitary vectors uniformly distributed on the Grassmann manifold, is proved to be asymptotically optimal under the criterion of rate maximization as the codebook size grows. On the other hand, for a finite codebook size and a MIMO channel with rich scattering, an important finding is reported in~\cite{mukkavilli2003beamforming, love2003grassmannian} that the optimal beamforming/precoding codebook design can be translated into the mathematical problem of Grassmannian line/subspace packing. The result, however, may not hold for correlated channels. This motivates a vein of research on developing systematic methods for Grassmannian codebook construction targeting spatially/temporally correlated channels, where the correlation is exploited for reducing the required codebook size and hence the feedback overhead~\cite{raghavan2007systematic,kim2011mimo,xia2006design,huang2009limited}. The latest research in the area is focused on next-generation massive MIMO communication with large-scale arrays. The resultant channels are high-dimensional and thus the limited feedback techniques developed previously for small-scale MIMO cannot be directly applied. The main approach for overcoming the challenge is to decompose a massive MIMO channel into the low-dimensional slow and fast time varying components, corresponding to channel spatial correlation and small-scale fading, respectively~\cite{nam2012joint,choi2013noncoherent,sim2016compressed}.  Then periodic feedback is needed only for a beamformer/precoder matched to the small-scale MIMO fading channels. The low-dimensional feedback can be then compressed using a traditional Grassmannian quantization technique,  thereby reining in feedback overhead. 

Stochastic gradient quantization for FEEL is related to limited feedback in that they both aim at compressing a vector for efficient transmission to convey directional information. To be specific, one critical information conveyed by a stochastic gradient is the descending direction on a surface generated by a given loss function, which measures the learning accuracy. Though Grassmannian quantization seems to be a suitable tool for gradient compression, the direct application is impractical as the codebook size and computation complexity both increase exponentially with the gradient's dimensionality. The techniques developed in the other area of limited feedback for massive MIMO with large-scale channels cannot be transferred to high-dimensional gradient quantization. The reason is that the former's effectiveness hinges on the channel's structural and multi-time-scale properties but no similar counterparts exist for stochastic gradients. This calls for the development of a new approach for high-dimensional gradient quantization.

\subsection{Contributions and Organization}
This work addresses the issue of practical quantization of high-dimensional stochastic gradients to realize communication-efficient FEEL in a wireless system. To this end, we propose a novel framework of hierarchical gradient quantization based on gradient decomposition and low-dimensional Grassmannian quantization. The effect of the framework on the learning performance is characterized by analyzing the model convergence rate as a function of the number of quantization bits, which quantifies the communication overhead. To the best of the authors' knowledge, this work represents the first attempt on applying Grassmannian quantization to the compression of high-dimensional stochastic gradients. 

The specific contributions of this work are summarized as follows. 

\begin{itemize}
\item \textbf{Hierarchical Quantization Architecture}: The practicality of the proposed framework arises from a hierarchical architecture for intelligent gradient decomposition and low-dimensional component quantization. First, the stochastic gradient is decomposed into its norm and the normalized stochastic gradient. The norm is easily compressed using a scalar quantizer. Next, as a key feature of the framework, the high-dimensional normalized stochastic gradient is intelligently decomposed into 1) a set of equal-length unitary vectors called normalized block gradients and 2) a mentioned hinge vector, which is also unitary and integrates normalized block gradients to yield the normalized stochastic gradient. Such decomposition has two advantages. The hinge vector harnesses a certain level of high-dimensional VQ gain even though only practical low-dimensional quantizers are deployed. The other advantage is that the unitary nature of normalized block gradients and hinge vectors allow them to be efficiently compressed using two Grassmannian quantizers.

\item \textbf{Bit-Allocation Scheme}: Under an overhead constraint, the total number of bits from quantizing a stochastic gradient is fixed. The allocation of the bits to control the resolutions of the quantizers in the proposed framework can be optimized under the criterion of minimum distortion, yielding a bit-allocation scheme. By average distortion analysis, as the length of block gradients increases, it can be proved that the randomness of the hinge vector diminishes and it converges to a known fixed point (a vector) on the Grassmann manifold, which is a vector and denoted as $\bh_{\sf ref}$. This suggests that in this asymptotic regime, all bits should be allocated to quantizing and transmitting normalized block gradients and the stochastic gradient norm with $\bh_{\sf ref}$ being used at the server as a surrogate for the hinge vector. Otherwise, in the non-asymptotic regime, the optimal bit-allocation is derived in closed-form by minimizing the sum distortion.

\item \textbf{Analysis of Learning Convergence Rate}: It is proved that the proposed hierarchical quantization scheme leads to the convergence of the FEEL algorithm even if the loss function is non-convex. The specific findings are three-fold: First, given a large number of edge devices, the convergence rate is asymptotically $O\l(\frac{1}{\sqrt{N}}\r)$ with $N$ denoting the total number of iterations; Second, the quantization error leads to a biased term on the upper bound of the expected gradient norm, where the increment of quantization bits reduces the value of this biased term, giving rise to a faster convergence speed; Third, given the quantization bits, this bias vanishes at the rate of $O(\frac{1}{K})$ with $K$ denoting the total number of edge devices, which also accelerates model convergence.

\end{itemize}
 
 \emph{Organization}: The remainder of the paper is organized as follows. Section~\ref{System model} introduces the FEEL system model and provides the problem formulation. Section~\ref{Hierarchical scheme} presents the hierarchical quantization scheme. The distortion analysis for the proposed scheme is given in Section~\ref{Distortion_analyses}, building on which, the bit-allocation strategy is derived in Section~\ref{Bit_allocation}. Section~\ref{Learning_rate_analysis} presents the convergence rate analysis of the learning algorithm with the proposed hierarchical quantization. Simulation results are provided in Section~\ref{simulation} followed by concluding remarks in Section~\ref{Conclusion}. 
 

\section{System Model and Problem Formulation}\label{System model}
Consider a FEEL system as illustrated in Fig.~\ref{Fig:system_model},
where an edge server trains an AI model (e.g., a classifier), represented by the parameter set $\boldsymbol\theta$, using training datasets distributed among $K$ edge devices. 

To facilitate the learning, the loss function measuring the model error is defined as follows. Let $\{\mathcal D_{k}\}$ denote the local dataset collected at the $k$-th edge device. The local loss function of the model vector $\boldsymbol\theta$ on $\{\mathcal D_{k}\}$ is given by
\begin{align}\label{LocalLoss}
\text{(Local loss function)}\quad
f_{k}(\boldsymbol\theta) = \frac{1}{|\mathcal D_{k}|}\sum_{(\bx_i,y_i) \in \mathcal D_{k}}f(\boldsymbol\theta, \bx_i,y_i),
\end{align}
where $f(\boldsymbol\theta, \bx_i,y_i)$ is the sample-wise loss function quantifying the prediction error of the model $\boldsymbol\theta$ on the training sample $\bx_i$ w.r.t its ground-true label $y_i$. Without loss of generality, by assuming uniform sizes for local datasets: $\{\mathcal D_{k}\} = D$, the global loss function of the model vector $\boldsymbol\theta$ on all distributed local datasets can be written as 
\begin{align}
\text{(Global loss function)}\quad
F(\boldsymbol\theta) = \frac{1}{K}\sum_{k=1}^{K}f_{k}(\boldsymbol\theta).
\end{align}
 
The learning process is to minimize the global loss function $F(\boldsymbol\theta)$, which can be mathematically written as 
\begin{equation}
\boldsymbol\theta^* = \arg\min F(\boldsymbol\theta).
\end{equation}

\begin{figure}[t]
  \centering
\includegraphics[width=0.75\textwidth]{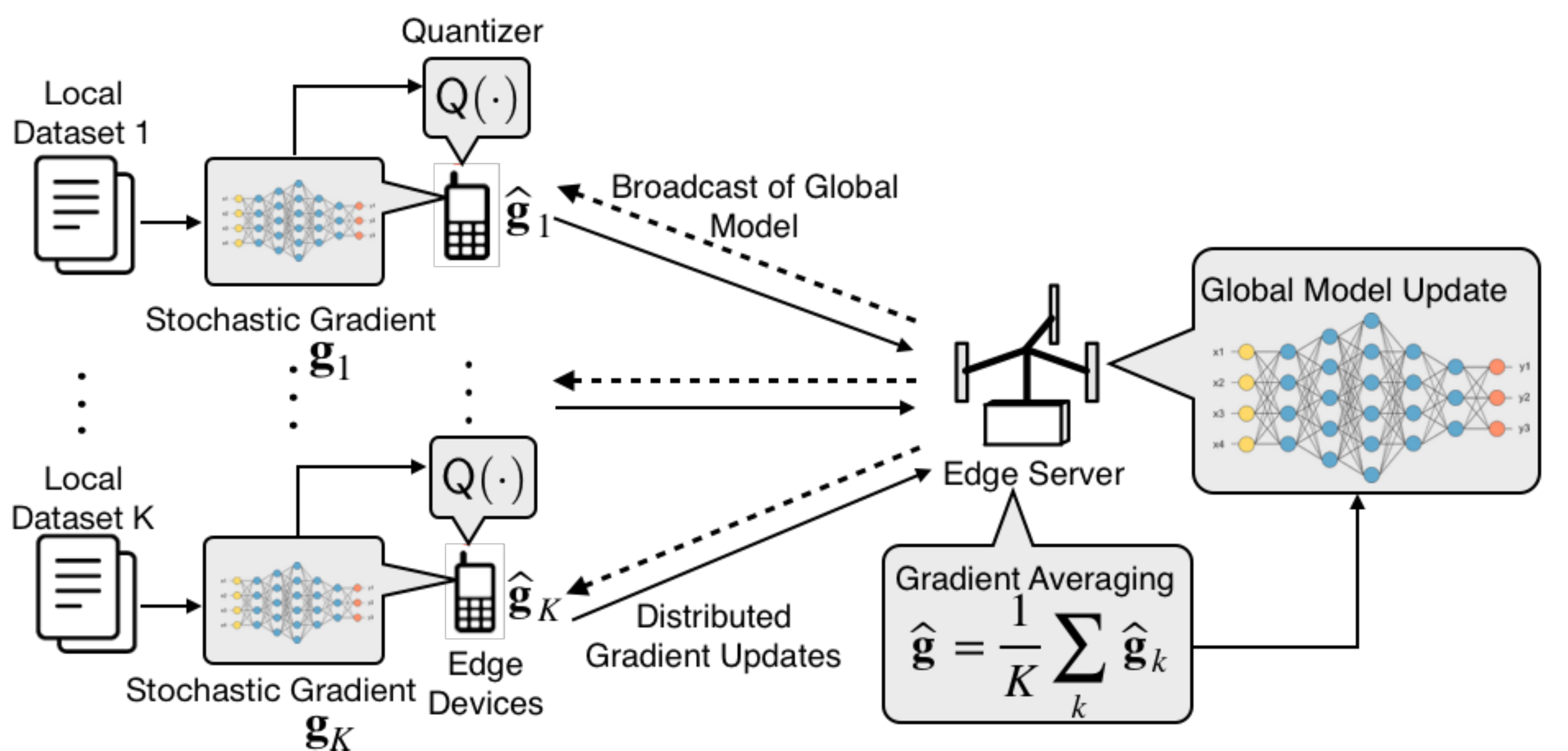} 
     \caption{Federated edge learning system.}
  \vspace{1mm}
  \label{Fig:system_model}
\end{figure}

In the context of FEEL, the gradient-averaging implementation is proposed in~\cite{konevcny2016federated} to tackle the privacy issue by avoiding uploading all the local data. Specifically, in each iteration, say the $n$-th iteration, the edge server broadcasts the current model under training $\boldsymbol\theta[n]$ to all edge devices. Based on the received current model $\boldsymbol\theta[n]$, each device computes the stochastic gradient by differentiating the local loss function defined in~\eqref{LocalLoss}. Mathematically, for device $k$, the stochastic gradient of the $n$-th iteration can be computed as:
\begin{equation}\label{StoGrad}
\text{(Stochastic gradient)}\quad \bg_k[n] = \nabla  f_{k}(\boldsymbol\theta[n]),
\end{equation}
where $\nabla$ denotes the gradient operation. Conventionally, upon its completion, the local gradients are sent to the edge server for averaging. However, in practical applications, communicating the gradients in each iteration has been observed to be a significant performance bottleneck~\cite{alistarh2017qsgd}, which will become exacerbated especially when the gradient is dense. This motivates the lossy compression of the gradients before transmission, mathematically defined as follows:
\begin{equation}\label{QuantSGD}
 \text{(Stochastic gradient quantization)}\quad \widehat{\bg}_k[n] = \Quan(\bg_k[n]),
\end{equation}
where $\Quan(\cdot)$ maps any point
$\bg$ in
$\mathbb R^{\sf Dim \times 1}$ to one of the codewords in the codebook $\mathcal C$, i.e. $\widehat{\bg}_k[n]\in\mathcal C$ with $\widehat{\bg}_k[n]$ denoting the quantized version of the stochastic gradient $\bg_k[n]$. Rather than conveying the quantized version, the edge devices communicate the codeword index to the edge server. It is further assumed that the edge server has the knowledge of the codebook $\mathcal C$ and the codeword index is perfectly transmitted. It means that the quantized version of the stochastic gradients can be perfectly transmitted. Then, by averaging all the quantized stochastic gradients, the approximated global gradient can be computed as:
\begin{equation}\label{ApproGloGrad}
\text{(Approximated Global gradient)}\quad \widehat{\bg}[n] = \frac{1}{K}\sum_{k = 1}^{K}\widehat{\bg}_k[n],
\end{equation}
where $\widehat{\bg}[n]$ denotes an estimate of the global gradient at the $n$-th iteration. Then, the global model $\boldsymbol\theta$ is updated as follows:
\begin{equation}\label{ModelUpdate}
\text{(Model updating)}\quad \boldsymbol\theta[n+1] = \boldsymbol\theta[n] - \eta\widehat{\bg}[n],
\end{equation}
where $\eta$ is the step size. The learning process involves the iteration from~\eqref{StoGrad} to~\eqref{ModelUpdate} until the model converges. 

Let us define the normalized stochastic gradient $\bff =
\frac{\bg}{\lVert\bg\rVert}$ and the norm of the stochastic gradient $\rho =
\lVert\bg\rVert$. For analytical tractability, we make the following
assumption. 
\begin{assumption}
  The normalized stochastic gradient $\bff$ is uniformly distributed on the
  Grassmann manifold. 
\end{assumption}%
In other word, we assume that $\bg$ is isotropic~(i.e., statistically
invariant under unitary transformation). To support our assumption, we have
run a hypothesis test on  
a real stochastic-gradient-dataset. Specifically, we apply the
\emph{Kolmogorov-Smirnov test} (KS-test), which is a nonparametric
hypothesis test quantifying a distance between the empirical
distribution function and the referenced one. Let the null hypothesis
$\mathcal H_0$ be such that $\bff$ is uniformly distributed, and 
$\mathcal H_1$ the alternative one. The obtained \emph{p-value}, which indicates
whether to reject or accept the null hypothesis $\mathcal H_0$, is close
to $0.1$. Given that the commonly used threshold for $p$-value is
$0.05$, this result suggests that the null hypothesis $\mathcal H_0$
will be accepted, i.e., the normalized stochastic gradient is believed to be uniformly distributed on the Grassmann manifold.

We are interested in the \emph{mean squared error}~(MSE) of the stochastic gradient
quantization problem. For a given \emph{quantization codebook}~$\mathcal C$ of
$B$~bits~(i.e., containing $2^B$ codewords), the optimal quantization function in the MSE sense is
such that $ \Quan_{\mathcal{C}}^E(\bg) \in
\arg\min_{\hat{\bg}\in\mathcal{C}} \| \bg - \hat{\bg} \|^2$. And
the distortion of the quantizer is denoted by $D_{\mathcal{C}}^E =
\mathbb{E}\left\{ \| \bg - \Quan_{\mathcal{C}}^E(\bg) \|^2 \right\}$. 
Here, the superscript  `E' stands for Euclidean distance. The optimal
codebook is therefore 
\begin{equation}
  \mathcal{C}^* \in \arg\min_{\mathcal{C}} D_{\mathcal{C}}^E. 
\end{equation}%

Two main challenges of the above optimal quantization problem are:
1)~codebook optimization which is NP hard;~2)~VQ which
is also NP hard with respect to the dimension for a general codebook.
Therefore, the goal of this work is to propose a hierarchical codebook
design which enables VQ with low complexity.

\section{Hierarchical Gradient Quantization }\label{Hierarchical scheme}

In this paper, we propose to quantize the gradient norm $\rho$ with a
scalar codebook $\mathcal{C}_\rho$ with $B_\rho$ bits and the normalized stochastic gradient
$\bff$ with a Grassmannian codebook\footnote{A Grassmannian codebook is
a set of unit norm codewords.} $\mathcal{C}_{\bff}$ with
$B_{\bff}$ bits. This is motivated by the suitability of such a codebook for quantizing 
a vector that contains directional information and the tractability of relevant designs~\cite{dhillon2008constructing}. 

Nevertheless, directly designing the codebook for $\bff$ is impractical due
to its high-dimensionality. 
To further reduce the complexity, we propose to decompose $\bff$ prior to quantization as
follows. Assuming that $\mathsf{Dim} = LM$ for some
integers\footnote{We can apply zero padding if $\mathsf{Dim} \ne LM$.}
$M$ and $L$, we partition the vector $\bff$ into $M$ blocks of
length $L$, i.e., $\bff^T = [\bv^T_1,\ldots,\bv^T_M]$. We call $\bv_i$ the
$i$-th \emph{block gradient}. Also, let us define the \emph{normalized
block gradient} $\bs_i = \frac{\bv_i}{\|\bv_i\|}$ and the
\emph{hinge vector} $\bh = [h_1,\ldots,h_M]^T$ where $h_i = \|\bv_i\|$,
$\forall\,i=1,\ldots,M$. It follows that both the normalized block gradients
and the hinge vector have unit norm and can be quantized with
Grassmannian quantizers. In addition, one can show that the normalized
block gradients are also isotropic, which is formally established in the following lemma.

\begin{lemma} \label{Lemma: Uniformity}
(Uniformity of normalized block gradients).
If $\bff = [f_1,f_2,\cdots,f_{\sf Dim}]^T$ is a uniformly distributed unitary random vector and given $\bv = [f_m,f_{m+1},\cdots,f_{n}]^T$ with $m<n$ an arbitrary block gradient picked from $\bff$, one can have that the normalized block gradient $\bs = \frac{\bv}{\lVert\bv\rVert}$ is uniformly distributed on the Grassmann manifold.
\proof
See Appendix~\ref{Proof: Lemma_uniformity}.
\endproof
\end{lemma}

Given the above decomposition and properties, we propose a quantization scheme with the following main ingredients, as shown in
Fig.~\ref{Fig:hierarchical_scheme}. 
\begin{itemize}
  \item For the stochastic gradient norm: $B_\rho$-bit scalar quantizer
    $\mathcal{C}_\rho$;
  \item For the normalized block gradients: $B_\bs$-bit uniform and
    even\footnote{We call $\mathcal{C}$ an even codebook if it can be partitioned as $\mathcal{C} = \mathcal{C}^+ \bigcup
    \mathcal{C}^-$ with $\mathcal{C}^+\bigcap\mathcal{C}^-=\emptyset$
    such that $-\bc\in\mathcal{C}^-$, $\forall\,\bc\in\mathcal{C}^+$.}
    Grassmannian quantizer $\mathcal{C}_\bs$;
\item For the hinge vector: $B_{\bh}$-bit positive\footnote{We call
 $\mathcal{C}$ a positive codebook if all codewords have only
 positive entries.} Grassmannian quantizer $\mathcal{C}_\bh$. 
\end{itemize}
\begin{figure}[t]
  \centering
\includegraphics[width=0.65\textwidth]{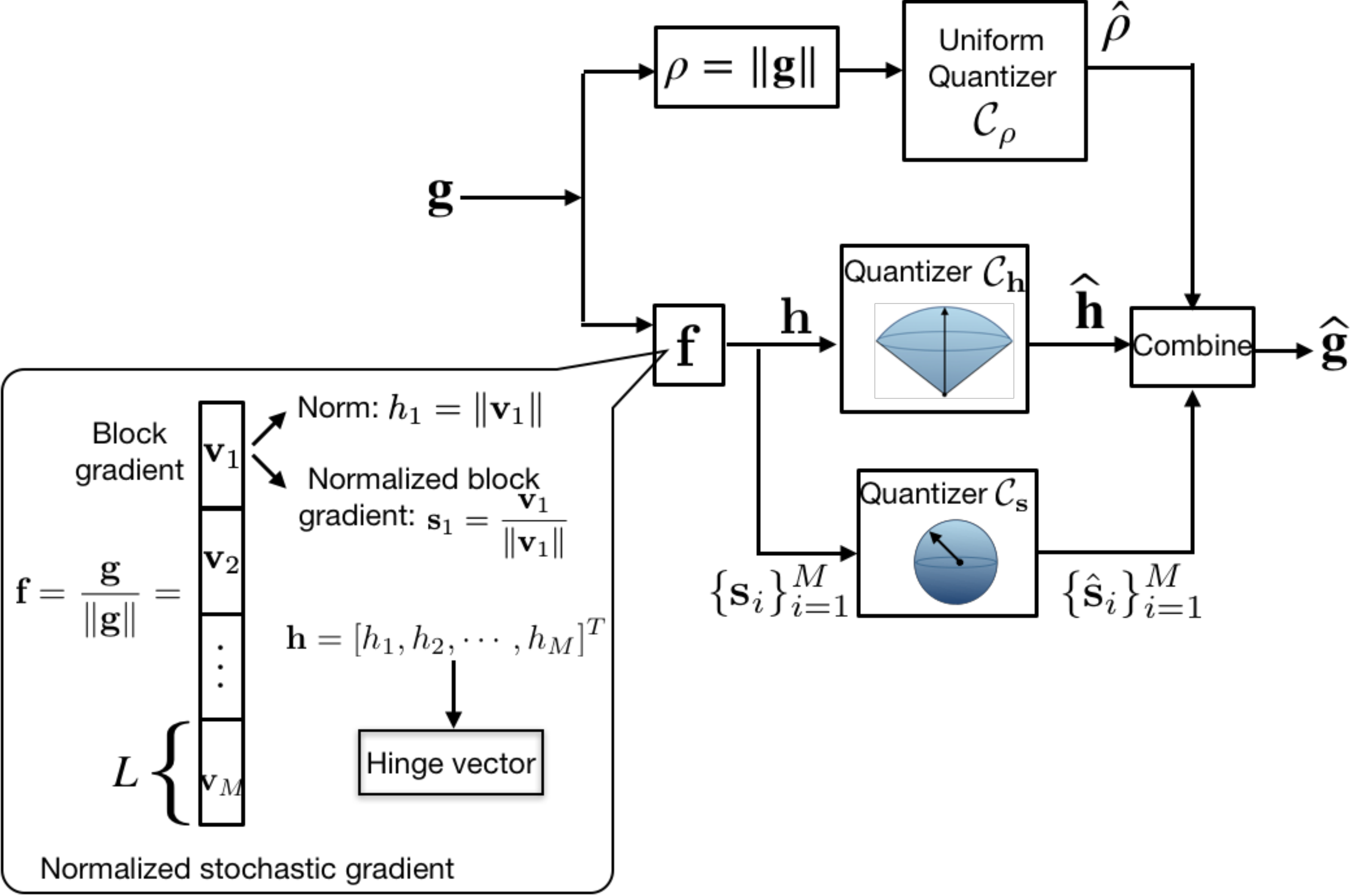} 
     \caption{Hierarchical quantization scheme.}
  \vspace{0mm}
  \label{Fig:hierarchical_scheme}
\end{figure}
In total, we need $B = B_\rho + M B_\bs + B_{\hv}$ bits. The quantized
version of $\bg$ is 
\begin{equation}
  \hat{\bg}^T=\hat{\rho}[\hat{\bff}^T_1,\ldots,\hat{\bff}^T_M] =
  \hat{\rho}[\hat{h}_1\hat{\bs}^T_1,\ldots,\hat{h}_M\hat{\bs}^T_M],
\end{equation}
where $\hat\rho$, $\hat{h}_i$, and $\hat{\bs}_i$ denote the quantized versions of $\rho$, ${h}_i$, and ${\bs}_i, \forall i$. 

For the above quantizers, we focus on quantization functions that minimize the Euclidean
distance~(MSE) between $\bg$ and $\hat\bg$. Let $\bx$ and $\hat{\bx}$ be two unit norm vectors, and let
  $d_c(\bx,\hat\bx) = \sqrt{1-|\hat\bx^T\bx|^2}$ be the
  chordal distance that measures angular deviation between $\bx$ and $\hat{\bx}$. 
 The two following lemmas are straightforward.

\begin{lemma}
  For a positive Grassmannian codebook~$\mathcal{C}$ and a given unitary
  vector $\bx$ with positive entries, if $\hat{\bx}\in\arg\min_{\hat{\bx}\in\mathcal{C}}
  d_c(\bx,\hat\bx)$, then
  $\hat{\bx}\in\arg\min_{\bx\in\mathcal{C}} \|\bx-\hat\bx\|^2$. 
\end{lemma}

\begin{lemma}
  For an even Grassmannian codebook~$\mathcal{C}$ and a given unit norm
  vector $\bx$, if $\hat{\bx}\in\arg\min_{\hat{\bx}\in\mathcal{C}^+}
  d_c(\bx,\hat\bx)$, then
  either $\hat{\bx}$ or $-\hat{\bx}$ belongs to
  $\arg\min_{\bx\in\mathcal{C}} \|\bx-\hat\bx\|^2$. 
\end{lemma}
Hence, to find a codeword in $\mathcal{C}$ with the shortest Euclidean
distance to $\bx$, we first find a codeword $\hat{\bx}$ in
$\mathcal{C}^+$ with the shortest
chordal distance to $\bx$. If the inner product between the pair is negative, then we flip the
sign of the codeword, which is a codeword in $\mathcal{C}^-$ and, therefore, is still inside $\mathcal{C}$. 
In addition, if $\hat{\bx}$ is the quantized version of $\bx$ in
an even codebook, then 
\begin{equation}
  \frac{1}{2}\|\bx - \hat{\bx}\|^2 = 1 - \sqrt{1-d_c^2(\bx,\hat\bx)}.
\end{equation}%

\newcommand{\Eg}{E_g}
\newcommand{\Dchor}[1]{D^{\mathrm{C}}_{#1}}
\newcommand{\Deucl}[1]{D^{\mathrm{E}}_{#1}}
\newcommand{\Cs}{\mathcal{C}_\bs}
\newcommand{\Crho}{\mathcal{C}_\rho}
\newcommand{\Ch}{\mathcal{C}_\bh}
\newcommand{\Cf}{\mathcal{C}_\bff}

\newcommand{\href}{\bh_{\text{ref}}}
\newcommand{\Bhref}{\mathcal{B}_{\href}}

\begin{proposition}
The mean square error of the proposed quantizer is 
\begin{equation}
  \mathbb{E}\left\{ \|\bg - \hat\bg \|^2 \right\} =
  \Deucl{\Crho} + \Eg \Deucl{\Cf},
\end{equation}%
where $\Eg = \mathbb{E}\left\{ \| \bg \|^2 \right\}$, and 
\begin{equation}
  \Deucl{\Cf} = 
  \Deucl{\Cs} + \Deucl{\Ch} - \frac{1}{2}
  \Deucl{\Cs} \Deucl{\Ch},
\end{equation}%
with $\Deucl{\Cs} =
\mathbb{E}\left\{ \| \bs - \Quan_{\mathcal{C}_{\bs}}^E(\bs) \|^2 \right\}$, $\Deucl{\Ch} =
\mathbb{E}\left\{ \| \bh - \Quan_{\mathcal{C}_{\bh}}^E(\bh) \|^2 \right\}$.
\end{proposition}
\begin{corollary}\label{cor:MSE}
  When $\max\{\Deucl{\Cs}, \Deucl{\Ch}\}$ is small,
  we have
  \begin{equation}
    \mathbb{E}\left\{ \|\bg - \hat\bg \|^2 \right\}
    \lesssim
    \Deucl{\Crho} + \Eg \left( \Deucl{\Cs} +
  \Deucl{\Ch}  \right), \label{eq:MSE}
\end{equation}%
where $\lesssim$ means the upper bound is asymptotically tight. 
\end{corollary}
There are two sub-problems for the codebook design problem. First, 
 for given bit allocation $(B_\rho, B_\bs, B_\bh)$, we need to jointly
 design the codebooks $(\Crho, \Cs,
 \Ch)$. From Corollary~\ref{cor:MSE}, it is 
asymptotically optimal, in the sense of scaling law, to design the codebooks $\Crho$,
$\Cs$, $\Ch$ separately. Second, we should
optimize the bit allocation such that the MSE~\eqref{eq:MSE} is
minimized. Given that it is hard to obtain tractable MSE expression as a
function of the codebook size, we investigate the more tractable expected squared
chordal distance instead. Such choice is justified by the closeness
between the two measures in the following sense
\begin{align}
  d_c^2(\bx,\hat\bx) &\le \|\bx - \hat{\bx}\|^2 \le 2d_c^2(\bx,\hat\bx). 
\end{align}%
Mathematically, we have the following optimization problem
\begin{align}
  \min_{B_\rho, B_\bs, B_\bh} &\left\{ 
   \min_{\Crho}  \Deucl{\Crho} + \Eg \left(
   \min_{\Cs} \Dchor{\Cs} +
  \min_{\Ch} \Dchor{\Ch}  \right) 
  \right\},\label{Eq: Distortion}\\
  &\text{s.t.}~B_\rho + M B_\bs + B_{\hv} \le B, 
\end{align}%
where $\Dchor{\mathcal{C}} = \mathbb{E}\left\{
d_c^2(\bx,\Quan_{\mathcal{C}}^C(\hat{\bx})) \right\}$ with the
superscript `C' for `{chordal}' distance.

\subsubsection{The design of Grassmannian codebook $\Cs$}\label{Block Gradient Quantization}
Recall that the codebook $\Cs$ is an even codebook. Hence, it is enough
to first construct a codebook $\Cs^+$ of size $2^{B_\bs-1}$, then construct $\Cs^- = \{
\bc:\,-\bc \in \Cs^+ \}$, and finally let $\Cs = \Cs^+ \bigcup \Cs^-$. The
optimal construction of $\Cs^+$ for isotropic sources is known as  
\emph{Grassmannian line packing}~\cite{love2003grassmannian},
formulated as 
\begin{align}
  (\text{codebook design for $\Cs^+$})\quad \max_{\Cs^+} \ \min\limits_{\hat{\bs}\ne\hat{\bs}'\in\Cs^+}
d_c\l(\hat{\bs},\hat{\bs}'\r). 
\end{align}%
\subsubsection{The design of Grassmannian codebook $\Ch$}\label{Hinge Vector Quantization}
An optimal codebook $\Ch$ should satisfy
\begin{equation}\label{SimpObj}
  \Ch \in \arg\min\limits_{\Ch}\ \mathbb E[d^2_c(\bh,\Quan_{\Ch}^C(\bh))]. 
\end{equation}%
Since the hinge vector is not isotropic, uniform quantization is not
optimal in general. A practical suboptimal solution is the \emph{Lloyd algorithm}~\cite{lau2004design}
on the Grassmann manifold that can be implemented iteratively.

\subsubsection{The design of scalar codebook $\Crho$} Similarly, 
an optimal codebook $\Crho$ should satisfy
\begin{equation}
  \Crho \in \arg\min\limits_{\Crho}\ \mathbb
  E[\|\rho-\Quan_{\Crho}^E(\rho)\|^2]. 
\end{equation}%

For simplicity, we adopt a uniform quantizer for $\rho = \lVert\bg\rVert$. So far, we have developed the hierarchical quantization scheme for stochastic gradients using three codebooks (see Fig.~\ref{Fig:hierarchical_scheme}). 

\section{Distortion Analysis}\label{Distortion_analyses}

In this section, the distortion from quantizing the stochastic gradients under the proposed hierarchical scheme will be analyzed, which involves the derivations of $ \Dchor{\Cs}$, $ \Dchor{\Ch}$ and $\Deucl{\Crho}$ in~\eqref{Eq: Distortion}, respectively.

\subsubsection{Distortion analysis on the normalized block gradient} According to~\cite{conway2013sphere,dai2008quantization}, the codebook designed by line packing is asymptotically optimal, and thus the resulting distortion for quantizing uniformly distributed unitary random vectors is asymptotically identical to that using the random codebook. Mathematically, the average distortion for the normalized block gradient is characterized as follows.

\begin{lemma}\label{Lemma:Block_Q-error}
(~\cite{dai2008quantization}, Theorem 2). Let $\Cs$ be a codebook
designed by line packing, with resolution $B_{\bs}-1$ and dimensionality
$L$. The average distortion, denoted as $ \Dchor{\Cs}$, incurred by quantizing the uniformly distributed unitary random vector under the codebook $\Cs$ can be bounded as
\begin{equation}\label{Eq:Block_Upper}
\frac{L-1}{L+1}2^{-\frac{2(B_{\bs}-1)}{L-1}} + o(1) \leq  \Dchor{\Cs}
\leq 2^{-\frac{2(B_{\bs}-1)}{L-1}} + o(1),
\end{equation} 
where $o(1)$ indicates a vanishing function of $L$ as $L \to \infty$.
\end{lemma}
One can observe from the above lemma that $\Dchor{\Cs}$ decreases exponentially as the codebook resolution $B_{\bs}$ increases. Moreover, fix $B_{\bs}$, as the length of the block increases, the quantization performance degrades accordingly in that pairwise distances between codewords enlarges according to the well-known results in line packing~\cite{conway2013sphere}.

\subsubsection{Geometric properties of the hinge vector}
To facilitate the derivation of $\Dchor{\Ch}$, the geometric properties of the hinge vector will be analyzed in this sub-section.

To begin with, we investigate the statistic distribution of the hinge vector by introducing the following lemma.
\begin{lemma}
If $X\sim\mathcal X^2\l(a\r)$ and $Y\sim\mathcal X^2\l(b\r)$ are
independent Chi-squared random variables, then $\frac{X}{X+Y}$ follows
the $\text{Beta}\l(\frac{a}{2},\frac{b}{2}\r)$ distribution. 
\end{lemma}

Since the normalized stochastic gradient $\bff$ is isotropic, thereby,
it can be further generated as $\bff = \frac{\bx}{\lVert\bx\rVert}$, where
the elements of $\bx$ are i.i.d. Gaussian random variables with zero
mean and unit variance. It follows that the square of each element of
the hinge vector $\bh = [h_1,h_2,\cdots,h_M]^T$ are beta distributed,
i.e.,  
\begin{equation}\label{Hinge_distribution}
 h^2_i = \frac{Z_1}{Z_1+Z_2} \sim\text{Beta}\l(\frac{L}{2},\frac{{\sf Dim}-L}{2}\r), \quad\forall i\in[1,M],
\end{equation}
where $Z_1\sim\mathcal X^2\l(L\r)$ and $Z_2\sim\mathcal X^2\l({\sf Dim}-L\r)$ are independent.

With this distribution at hand, the geometric properties of the hinge vector will be analyzed. 
 One specific result is characterized as below. 
\begin{proposition}[Convergence of hinge vector]\label{RefVector}\label{Proposition:hyper-sperical_cap}
For sufficiently large $L$, $\bh$ converges to a constant vector
$\bh_{\sf ref}= \frac{1}{\sqrt{M}} \pmb{1}_{M\times1}$. 
In particular, given $r = \sqrt{\frac{2}{1+ 2L^{\frac{1}{4}}}}$,
we have 
\begin{equation}
\emph{\text{Pr}}\l(d_c(\bh_{\sf ref},\bh)>r\r)< L^{-\frac{1}{2}} +
 O(L^{-\frac{3}{2}}), \label{eq:tmp821}
\end{equation}
\proof
The convergence is implied by \eqref{eq:tmp821} that is proved in Appendix~\ref{Proof:hyper-sperical_cap}.
\endproof
\end{proposition}


\begin{remark}
\emph{(Geometric Interpretation). According to
Proposition~\ref{Proposition:hyper-sperical_cap}, the hinge vector
converges to $\bh_{\sf ref}$ at least geometrically fast in block
dimensionality $L$. To be specific, the hinge vector locates with high probability within a ball of radius $r$ (with respect to
the chordal distance) on the Grassmann manifold centered at
$\href$, namely, $\Bhref(r) = \{\bh \;| \;\bh^T\bh = 1, d_c(\bh_{\sf ref},\bh)\leq r\}$, where $r$ converges to zero as $L$ grows.
}
\end{remark}

\subsubsection{Distortion analysis on the hinge vector} 

Based on the above analyses, we are now ready to derive the average distortion for quantizing the hinge vector.
To this end, we first introduce the following lemma.

\begin{lemma} \label{lemma:ball}
  If the set \emph{$\mathcal{C'} = \mathcal{C}\bigcap\Bhref( (1+\alpha)r)$} is
  not empty for some $\alpha>0$, then, for any \emph{$\bx \in \Bhref(r)$}, we have
  \begin{equation}
    \min_{\bc\in\mathcal{C}'} d_c(\bc, \bx) \le
    \bigl(1+\frac{2}{\alpha}\bigr) \min_{\bc\in \mathcal{C}}
    d_c(\bc, \bx) \label{eq:ball}
  \end{equation}%
\end{lemma}
\begin{proof}
  Let $\bc\in\mathcal{C}'$ and
  $\bc'\in\mathcal{C}\setminus\mathcal{C'}$. Thus, we have $d_c(\bc,\href)\le
  (1+\alpha)r$ and $d_c(\bc',\href) \ge (1+\alpha)r$. By the triangle
  inequality, we have 
  \begin{align}
     d_c(\bc, \bx) &\le d_c(\href,\bx) + d_c(\bc,\href) \le (2+\alpha) r, \\ 
     d_c(\bc',\bx) &\ge d_c(\bc', \href) - d_c(\bx,\href) \ge \alpha r, 
  \end{align}%
  from which we have 
  \begin{equation}
    d_c(\bc, \bx) \le \bigl(1+\frac{2}{\alpha}\bigr) d_c(\bc',\bx). 
  \end{equation}
  Taking the minimum on both sides, \eqref{eq:ball} is straightforward.  
\end{proof}

Let us now construct a codebook of $N = 2^{B_\bh}$ codewords as follows. 
First, we draw $N'\ge N$ points uniformly from the Grassmann manifold as
for the normalized block gradient codebook. Then, we choose the $N$ codewords that
are closest to $\href$ to form the codebook $\mathcal{C}_{\bh}$. To
analyze the average distortion, we introduce two balls $\Bhref(r)$ and $\Bhref(
(1+\alpha) r)$ for some $\alpha>0$ and $r>0$ that can be optimized later
on. Let us consider the following encoding rules.
\begin{itemize}
  \item If $\bh$ lies outside of $\Bhref(r)$, an encoding
    error is declared. This event has probability $P_1(r)$. 
  \item If no codeword lies inside $\Bhref( (1+\alpha)r)$, an encoding
    error is declared. This event has probability $P_2(\alpha,r)$.
  \item If there are more than $N$ codewords inside $\Bhref(
    (1+\alpha)r)$, an error is declared. This event has probability $P_3(\alpha,r)$.
\end{itemize}
We can upper-bound the distortion by $1$ whenever an error is declared,
then we have the following upper bound on the average distortion from
the union bound:
\begin{align}\label{Union_Bound}
  P_1(r) + P_2(\alpha,r) + P_3(\alpha,r) + (1-P_1(r))\left(
  1+\frac{2}{\alpha} \right) D(N'),
\end{align}%
where the last term is from Lemma~\ref{lemma:ball} and $D(N')$ is the
average distortion for a uniform random quantizer with $N'$ codewords. 

In particular, $P_3(\alpha,r)$ is the complementary cumulative
distribution function of a binomial distribution with parameter
$N'$ and $p$ where $p$ is the probability that a uniformly distributed
point falls inside the ball.

\begin{lemma} \label{Q-error_Hinge}
The distortion for quantizing the hinge vector can be upper-bounded as
\begin{align}\label{Upper_hinge}
\Dchor{\Ch} \leq L^{-\frac{1}{2}}(\beta_L2^{-\frac{2B_\bh}{M-1}} + 1) + O(L^{-\frac{3}{2}}),
\end{align}
where $\beta_L = \frac{1+r^{\frac{1}{2}}}{1-r^{\frac{1}{2}}}$ with $r = \sqrt{\frac{2}{1+ 2L^{\frac{1}{4}}}} \approx L^{-\frac{1}{8}}$.
\proof
See Appendix~\ref{Proof: Q-error_Hinge}.
\endproof
\end{lemma}

Several observations can be made from the above lemma. First, as the codebook resolution $B_{\hv}$ increases, the upper bound of the quantization error reduces accordingly due to the reduction of pairwise codewords distance. Second, given $B_{\hv}$, the upper bound is a decreasing function of the block length $L$. This is because the hinge vector tends to converge to $\bh_{\sf ref}$ given larger $L$ and the resulting quantization error reduces.
\subsubsection{Distortion analysis on the stochastic-gradient norm} given a uniform quantizer for the norm of stochastic gradients, the average distortion can be upper-bounded 
as
\begin{align}\label{Distortion: Norm}
\Deucl{\Crho} \leq \l(\frac{\Delta^{\sf quant}}{2}\r)^2,
\end{align}
where $\Delta^{\sf quant}$ denotes the quantization interval of a uniform quantizer.
\section{Quantization Bit Allocation}\label{Bit_allocation}
In this section, a practical \emph{bit-allocation scheme} will be
developed. Specifically, given a fixed number $B$ of bits, for quantizing the stochastic gradient vector, we aim to determine the scheme on how to allocate these bits to the three codebooks derived in the preceding section.

Given the proposed hierarchical quantization scheme, the original bit-allocation problem is formulated as 
\begin{align}
  \min_{B_\rho, B_\bs, B_\bh} &\left\{ 
   \Deucl{\Crho} + \Eg \left(
    \Dchor{\Cs} +
   \Dchor{\Ch}  \right) 
  \right\},\label{Eq: Distortion2}\\
  &\text{s.t.}~B_\rho + M B_\bs + B_{\hv} = B.
\end{align}%
Here, the distortions are replaced by the upper bounds derived
previously.  

For the \emph{Bit-allocation problem}, we assume for tractability that the
elements of stochastic gradient $\bg$ are i.i.d.~Gaussian 
with zero mean and unit variance. Note that this is also the worst case in the
sense that for a given variance the differential entropy is maximized
with Gaussian distributions. Then, it follows that $\Eg = ML$ and $\rho
= \lVert\bg\rVert \sim \mathcal X(ML)$. Given $\rho\in
[0,\rho_{\max}]$\footnote{Due to the fact that for $\rho \sim \mathcal
X(ML)$, $\mathbb E[\rho] + \sqrt{\text{Var[$\rho$]}} \le \mathbb E[\rho^2] + \sqrt{\text{Var[$\rho^2$]}}$ with $\sqrt{\text{Var[$\cdot$]}}$ denoting the standard deviation, for tractability, we take $\rho_{\max} = \mathbb E[\rho^2] + \sqrt{\text{Var[$\rho^2$]}} = ML + \sqrt{2ML}$.}, it follows from~\eqref{Distortion: Norm} that $\Deucl{\Crho}  \leq \frac{\rho^2_{\max}}{(2^{B_{\rho}+1}+2)^2}\leq \frac{\rho^2_{\max}}{4}2^{-2B_\rho}$. Then, the original bit-allocation problem can
be relaxed as $(\bf{P'})$, which is given at the top of this page.
\begin{figure*}
\begin{align}
(\bf{P'})\quad  \min_{B_\rho, B_\bs, B_\bh} &\left\{ 
    \frac{\rho^2_{\max}}{4}2^{-2B_\rho} + \Eg 2^{-\frac{2(B_\bs-1)}{L-1}}
    + \Eg \beta_L2^{-\frac{2B_\bh}{M-1}}L^{-\frac{1}{2}}
  \right\},\label{Eq: Distortion3}\\
  &\text{s.t.}~B_\rho + M B_\bs + B_{\hv} = B, \label{Bit_constaint}
\end{align}
\hrule
\end{figure*}

The above problem $(\bf{P'})$ is convex, and the optimal solutions can be derived by leveraging \emph{Karush-Kuhn-Tucker} (KKT) conditions 
as follows
\begin{align}
(\text{KKT conditions}) \left\{ 
 \begin{array}{rcl}
\lambda^*+\frac{\partial f(B^*_\rho,B^*_\bs,B^*_\bh)}{\partial B^*_\rho} = 0 \\
M\lambda^* + \frac{\partial f(B^*_\rho,B^*_\bs,B^*_\bh)}{\partial B^*_\bs} = 0\\
\lambda^* + \frac{\partial f(B^*_\rho,B^*_\bs,B^*_\bh)}{\partial B^*_\bh} = 0,
\end{array}
\right. 
\end{align}
where $f(B^*_\rho,B^*_\bs,B^*_\bh)$ is the objective of the optimization
problem, i.e.~\eqref{Eq: Distortion3}, and $\lambda^*$ is the Lagrange multiplier. Solving the 
above equations, we obtain the following bit-allocation scheme.
\begin{framed}
\begin{scheme}\emph{(Quantization Bit Allocation). To minimize the distortion, the bits $B$ can be allocated to the three quantizers as follows:
\begin{align}
B^*_\rho & = \lfloor \log_2 \frac{ML+\sqrt{2ML}}{2}+ \frac{1}{2}\log_2\ln 2+\frac{1}{2}-\frac{1}{2}\log_2\lambda^*\rfloor,\label{Bit: norm}\\
B^*_\bs & = \lfloor\frac{L-1}{2}\log_2\frac{2L}{L-1}+1 + \frac{L-1}{2}\log_2\ln 2-\frac{L-1}{2}\log_2 \lambda^*\rfloor ,\label{Bit: block}\\
B^*_\bh & = \lfloor\frac{M-1}{2}\log_2\frac{2}{M-1} + \frac{M-1}{2}\log_2\beta_LM\sqrt{L}\nn\\
& + \frac{M-1}{2}\log_2\ln 2-\frac{M-1}{2}\log_2\lambda^* \rfloor,\label{Bit: hinge}
\end{align}
where $\beta_L = \frac{1+ L^{-\frac{1}{16}}}{1-L^{-\frac{1}{16}}}$ and $\lambda^*$ can
 be obtained by substituting the above equations into~\eqref{Bit_constaint} as 
\begin{align}\label{Eq:lambda}
\log_2 \lambda^* = \frac{2}{ML}\log_2 \frac{ML+\sqrt{2ML}}{2} +\frac{L-1}{L}\log_2\frac{2L}{L-1}+ \frac{2}{L}+\frac{1}{ML} \nn\\
+ \frac{2(M-1)}{ML}\log_2\frac{2}{M-1}+ \frac{M-1}{ML}\log_2\beta_LM\sqrt{L} + \log_2\ln 2- \frac{2B}{ML}.
\end{align}
}
\end{scheme}
\end{framed}

Next, it is necessary to show that the above bit-allocation scheme is optimal in the sense of scaling law. To this end, we bound $\mathbb{E}\left[\lVert\bg-\widehat{\bg}\rVert^2\right]$ in the following theorem.
\begin{theorem}\label{Theorem: Lower_Upper_Bound}\emph{(Optimality of Bit Allocation).}
For sufficiently large $L$ and at the low resolution regime\footnote{The term `low resolution' is declared in the sense that only fewer than one bit is exploited for quantizing each coefficient of the stochastic gradient. This is a popular regime being explored in the area of edge learning~\cite{2018signsgd}.}, i.e. $B\leq ML$, the distortion under MSE metric per dimension, i.e. $\frac{\mathbb{E}\left[\lVert\bg-\widehat{\bg}\rVert^2\right]}{ML}$ can be bounded as 
\begin{equation}\label{Ineq: Upper_lower}
-\frac{2\ln 2}{ML}B \leq \ln \frac{\mathbb{E}\left[\lVert\bg-\widehat{\bg}\rVert^2\right]}{ML} \leq c_{\sf gap} - \frac{2 \ln 2}{ML}B + O\l({L^{-\frac{3}{2}}}\r),
\end{equation}
where $c_{\sf gap} =  \ln 2-\ln \frac{2L}{L-1} +\frac{2}{ML}\ln \frac{ML+\sqrt{ML}}{2}+\frac{L-1}{L}\ln \frac{2L}{L-1}+\frac{2}{L}\ln 2+\frac{2(M-1)}{ML}\ln \frac{2}{M-1}+ \frac{\ln 2}{ML}+\frac{M-1}{ML}\ln \beta_LM\sqrt{L} + 2(\beta_L+1)L^{-\frac{1}{2}} -2(\beta_L+1)^2L^{-1}$ with $\beta_L = \frac{1+ L^{-\frac{1}{16}}}{1-L^{-\frac{1}{16}}}$; $B$ denotes the number of bits used for quantizing the stochastic gradient $\bg$.
\end{theorem}

It can be observed from~\eqref{Ineq: Upper_lower} that the scaling law of the upper bound is the same as that of the lower bound with respect to $B$. It is further noted that the upper bound in the above theorem is derived by setting $B^*_{\bh}=0$. It means that the derived scaling law is independent of the number of bits allocated for quantizing the hinge vector. This is aligned with the intuition that as block length $L$ increases, the hinge vector tends to converge to the constant vector and the corresponding distortion is close to zero. 
From more theoretic point of view, it follows from~\eqref{Upper_hinge}
that, at the low resolution regime, no matter how many bits are
allocated to $\Ch$, the decaying-rate of the average distortion is
asymptotically bounded by $L^{-\frac{1}{2}}$. Thereby, this motivates a
practical bit-allocation scheme as given below. 

\begin{framed}
\begin{scheme}\emph{(Practical Bit Allocation). For the high-dimensional
  stochastic gradient $\bg$, $ B^*_\rho = \l\lfloor \log_2
  \frac{ML+\sqrt{2ML}}{2}+\frac{1}{2}\log_2\ln 2+\frac{1}{2}-\frac{1}{2}\log_2\lambda^* \r\rfloor$ bits are
  allocated for quantizing its norm with $\lambda^*$ is defined
  in~\eqref{Eq:lambda}. All rest bits should be allocated to the
  codebook $\Cs$ while exploiting $\bh_{\sf ref}=
  \frac{1}{\sqrt{M}}\pmb{1}_{M\times1}$ as a surrogate for the hinge vector.}
\end{scheme}
\end{framed}

\begin{remark}\emph{
The above practical bit-allocation scheme makes the proposed hierarchical quantization scheme be of low-complexity. To be specific, the relative low-dimensional block gradients makes the design complexity of codebook $\Cs$ via line packing algorithm reduces significantly compared to quantizing the high-dimensional stochastic gradient as a whole. On the other hand, the design complexity is further reduced without constructing the codebook $\Ch$ for the hinge vector.
}
\end{remark}

\section{Learning Convergence Rate Analysis}\label{Learning_rate_analysis}
Given a typical quantization scheme for the stochastic gradient, one concern related is that whether it will lead to the convergence of the learning algorithm. Thereby, in this section, the convergence rate of the learning algorithm under the proposed hierarchical quantization scheme will be theoretically investigated.

We begin our analysis in the non-convex setting, where we follow the standard assumptions of the stochastic optimization literature (see e.g.,~\cite{2018signsgd}). The specific assumptions are given as follows.

\begin{assumption}\label{Assump:Lower_bound}
\emph{(Lower Bound).} For all $\boldsymbol\theta$ and some constant $F^*$, we have that the global objective value $F(\boldsymbol\theta)\geq F^*$.
\end{assumption}

\begin{assumption}\label{Assump:smoothness}
\emph{(Smoothness).} Let $\bar{\bg}(\boldsymbol\theta)$ denote the gradient of the global objective $F(\boldsymbol\theta)$ evaluated at point $\boldsymbol\theta = \l[\theta_1,\theta_2,\cdots,\theta_{\sf Dim}\r]^T$ with ${\sf Dim} = ML$. Then $\forall \boldsymbol\theta$ and $\boldsymbol\beta = \l[\beta_1,\beta_2,\cdots,\beta_{\sf Dim}\r]^T$, we require that for some non-negative constant vector $\bl = \l[l_1,l_2,\cdots,l_{\sf Dim}\r]^T$ 
\begin{equation}\l|F(\boldsymbol\beta)-\l[F(\boldsymbol\theta)+\bar{\bg}(\boldsymbol\theta)^T\l(\boldsymbol\beta-\boldsymbol\theta\r)\r]\r|\leq \frac{1}{2}\sum_{i+1}^{\sf Dim}l_i(\beta_i-\theta_i).
\end{equation}
\end{assumption}

\begin{assumption}\label{Assump:variance_bound}
\emph{(Variance Bound).} The stochastic gradient $\bg(\boldsymbol\theta)$ is unbiased that has coordinate bounded variance:
\begin{equation}
\mathbb E[\bg(\boldsymbol\theta)] = \bar{\bg}(\boldsymbol\theta)\quad\text{and}\quad \mathbb E\l[\l(\bg(\boldsymbol\theta)_i-\bar{\bg}(\boldsymbol\theta)_i\r)^2\r]\leq \sigma^2_i,
\end{equation}
for a vector of non-negative constants $\boldsymbol\sigma = [\sigma_1,\sigma_2,\cdots,\sigma_{\sf Dim}]^T$.
\end{assumption}

Under the above three standard assumptions, we have the following result.

\begin{theorem}\label{Convergence}\emph{(Learning Convergence Rate).}
  Let $N$ be the number of iterations for the federated learning
  algorithm, $K$ the total number of users, and $\eta = \frac{1}{\sqrt{l_0N}}$ the learning rate. It follows that
\begin{align}\label{Eq:learning_convergence}
\mathbb E\l[\frac{1}{N}\sum_{n=0}^{N-1}\lVert\bar{\bg}_n\rVert^2\r] 
\leq 
\frac{\sqrt{l_0}\l(\frac{1}{2K}{\mathbb E}\l[\l\lVert\bg-\widehat\bg\r\rVert^2\r] +\frac{\lVert\boldsymbol\sigma\rVert^2}{2K}+F_0-F^*\r)}{\sqrt{N}-\frac{\sqrt{l_0}}{2K}},
\end{align}
where $F_0$ is the initial objective value and $F^*$ is defined in Assumption~\ref{Assump:Lower_bound}; $l_0 = \lVert\bl\rVert_{\infty}$ with $\bl$ defined in Assumption~\ref{Assump:variance_bound}.
\proof
See Appendix~\ref{Proof: Convergence}.
\endproof
\end{theorem}

Several observations can be made from~\eqref{Eq:learning_convergence} as
follows. First, the increment of the total iteration number $N$ leads to
the convergence of the learning algorithm. Specifically, as the number
of users $K\to \infty$, the convergence rate is asymptotically
$O\l(\frac{1}{\sqrt{N}}\r)$. Furthermore, as the number of users $K$
increases, the upper bound in~\eqref{Eq:learning_convergence} decreases.
This is because that the participation of more users, called multi-user gain, makes the
aggregated-and-averaged stochastic gradient closer to the true
gradient, leading to a faster convergence speed.

\section{Simulation results}\label{simulation}
Consider a FEEL system with one edge server and $K = 100$ edge devices.
The simulation settings are given as follows unless specified otherwise.
We consider the learning task of handwritten-digit recognition using the well-known MNIST dataset that consists of $10$ categories ranging from digit ``$0$" to ``$9$" and a total of $60000$ labeled training data samples.
The classifier model is implemented using a $6$-layer \emph{convolutional neural network} (CNN) that consists of two $5\times5$ convolution layers with ReLu activation (the first with $32$ channels, the second with $64$). Each followed with a $2\times2$ max pooling, a fully connected layer with $512$ units,  ReLu activation, and a final softmax output layer.  
Furthermore, it is noted that the total number of bits used for quantizing each coefficient of the stochastic gradients is $\frac{B_{\bs}}{L}+{\frac{B_{\rho}}{\sf Dim}}$ given $B_{\bh} = 0$ in the proposed bit-allocation scheme. Due to the fact that $\frac{B_{\rho}}{\sf Dim} = 0,{\sf Dim} \to \infty$, we define the number of bits per coefficient as $\frac{B_{\bs}}{L}$ without loss of generality.

\begin{figure}[t]
  \centering
\includegraphics[width=0.55\textwidth]{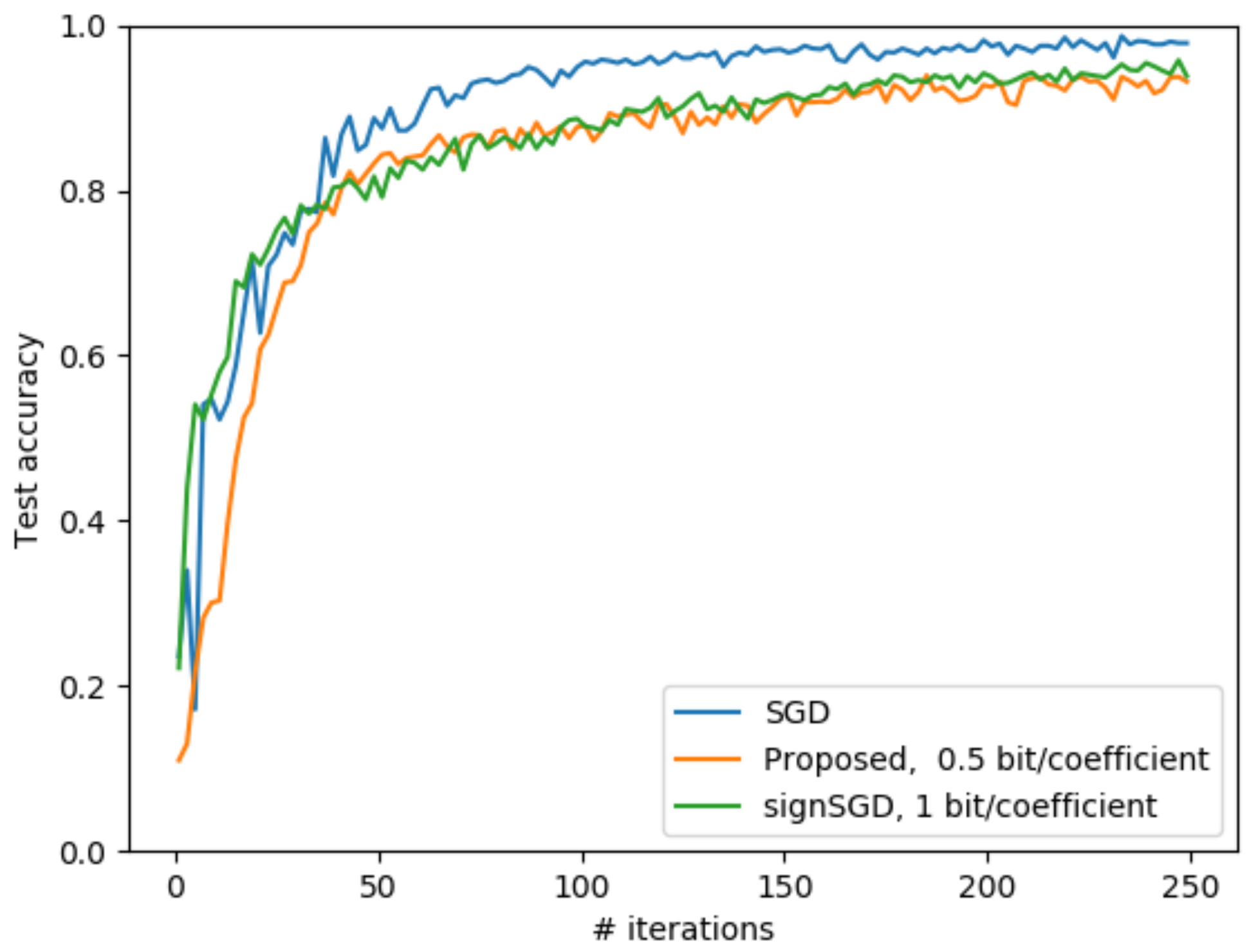} 
     \caption{Performance comparison of \emph{signSGD} and the proposed scheme.}
  \vspace{1mm}
  \label{Fig:performacne_comparison}
\end{figure}

\subsection{Performance of the Hierarchical Quantization Scheme}
The effectiveness of the proposed hierarchical quantization scheme is evaluated by benchmarking against \emph{signSGD} and SGD. The curves of the test accuracy versus the number iterations are illustrated in Fig.~\ref{Fig:performacne_comparison}. Several observations can be made as follows. First, using fewer bits, i.e. $0.5$ bit/coefficient, the performance of the proposed scheme is comparable to state-of-the-art \emph{signSGD}, which uses $1$ bit/coefficient. This attributes to the superiority of vector quantization over the scalar counterpart given the same bits at the low resolution regime. Furthermore, it can also be observed that SGD outperforms both quantization schemes because there exists quantization loss for both quantization schemes.

\begin{figure}[t]
  \centering
\includegraphics[width=0.55\textwidth]{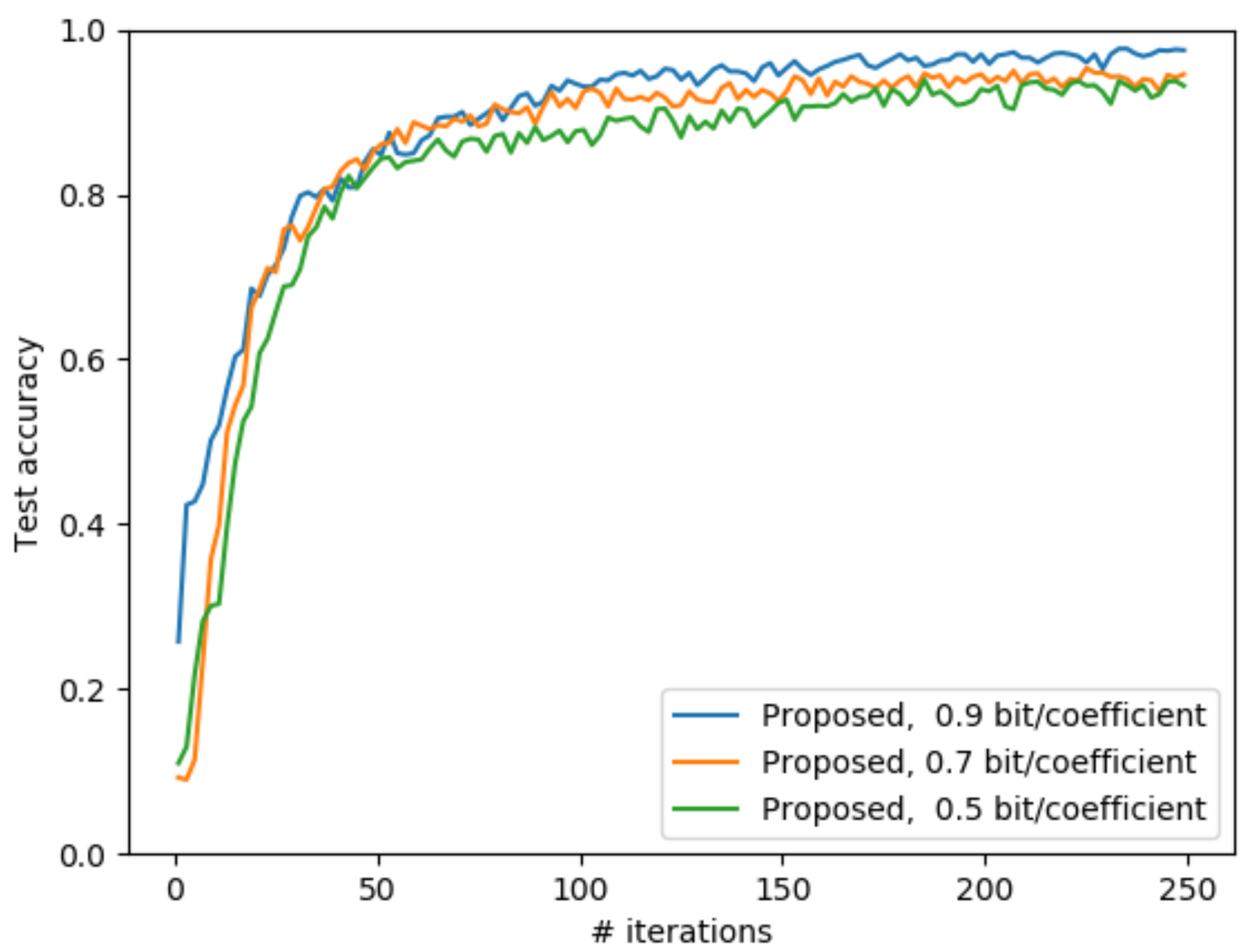} 
     \caption{Effect of the codebook resolution $B_{\bs}$ with block length $L  = 10$ and $B_{\rho} = 26$ bits.}
  \vspace{1mm}
  \label{Fig:effect_resolution}
\end{figure}

\subsection{Effect of the Codebook Resolution}
Given the block length $L$, the effect of the codebook resolution $B_{\bs}$ for $\Cs$ is evaluated, where the resolution for $C_{\rho}$ is fixed. The curves of test accuracy versus the number of iterations by varying the codebook resolution are illustrated in Fig.~\ref{Fig:effect_resolution}. It can be observed that as the codebook resolution increases, the learning performance improves accordingly. This is because the increment of resolution reduces the pairwise chordal distance among codewords. Then, the resulting quantization error reduces, giving rise to a better learning performance.

\begin{figure}[t]
  \centering
\includegraphics[width=0.55\textwidth]{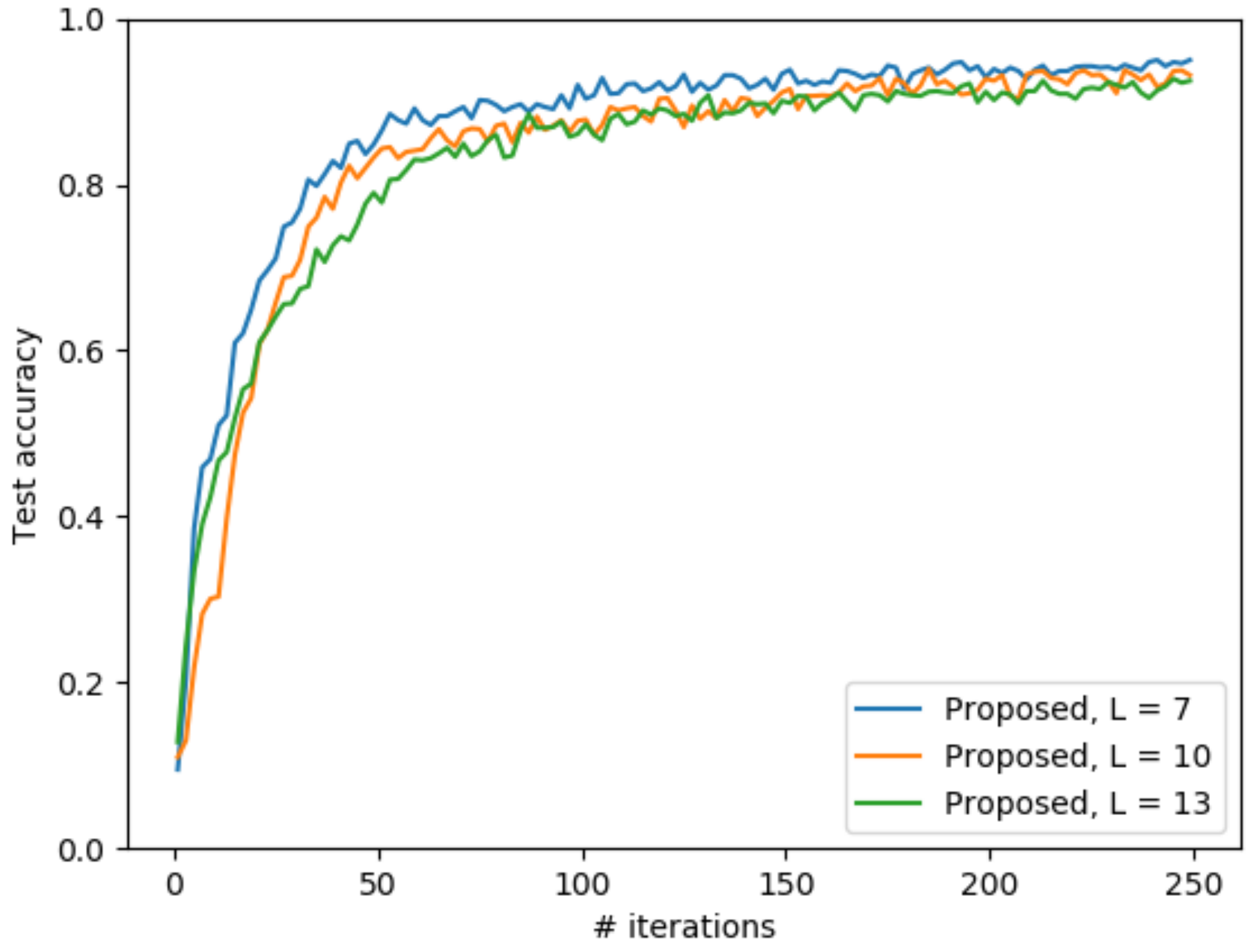} 
     \caption{Effect of the block length $L$ with $B_{\bs} = 5$ bits and $B_{\rho}  = 26$ bits. }
  \vspace{1mm}
  \label{Fig:effect_L}
\end{figure}

\subsection{Effect of the Block Length}
Given the fixed number of bits allocated to each block, the effects of the block length $L$ is evaluated. In particular, the curves of test accuracy versus the number of iterations by varying the block length $L$ are illustrated in Fig.~\ref{Fig:effect_L}. It can be observed that as $L$ increases, the learning performance degrades accordingly in that the quantization error for the stochastic gradients enlarges. The underlying reason is that a larger $L$ implies that the Grassmnnian codebook is generated by the packing algorithm on a higher dimensional Grassmann manifold. This enlarges the pairwise chordal distance between codewords and thus the resulting quantization error is large.

\subsection{Effect of the Edge-Device Number}
Fix the iteration number as $50$, the relationship between the learning
performance and the total number of edge-devices $K$ with various block
length $L$ is illustrated in Fig.~\ref{Fig:effect_user_number}. It can
be observed that as $K$ increases, the learning performance improves
accordingly. This is consistent with the result derived in
Theorem~\ref{Convergence}. Specifically, as indicated
by~\eqref{Eq:learning_convergence}, a larger $K$ reduces the noise
variance, and also makes the aggregated-and averaged stochastic gradient
closer to the true gradient, giving rise to a faster convergence speed.

\begin{figure}[t]
  \centering
\includegraphics[width=0.55\textwidth]{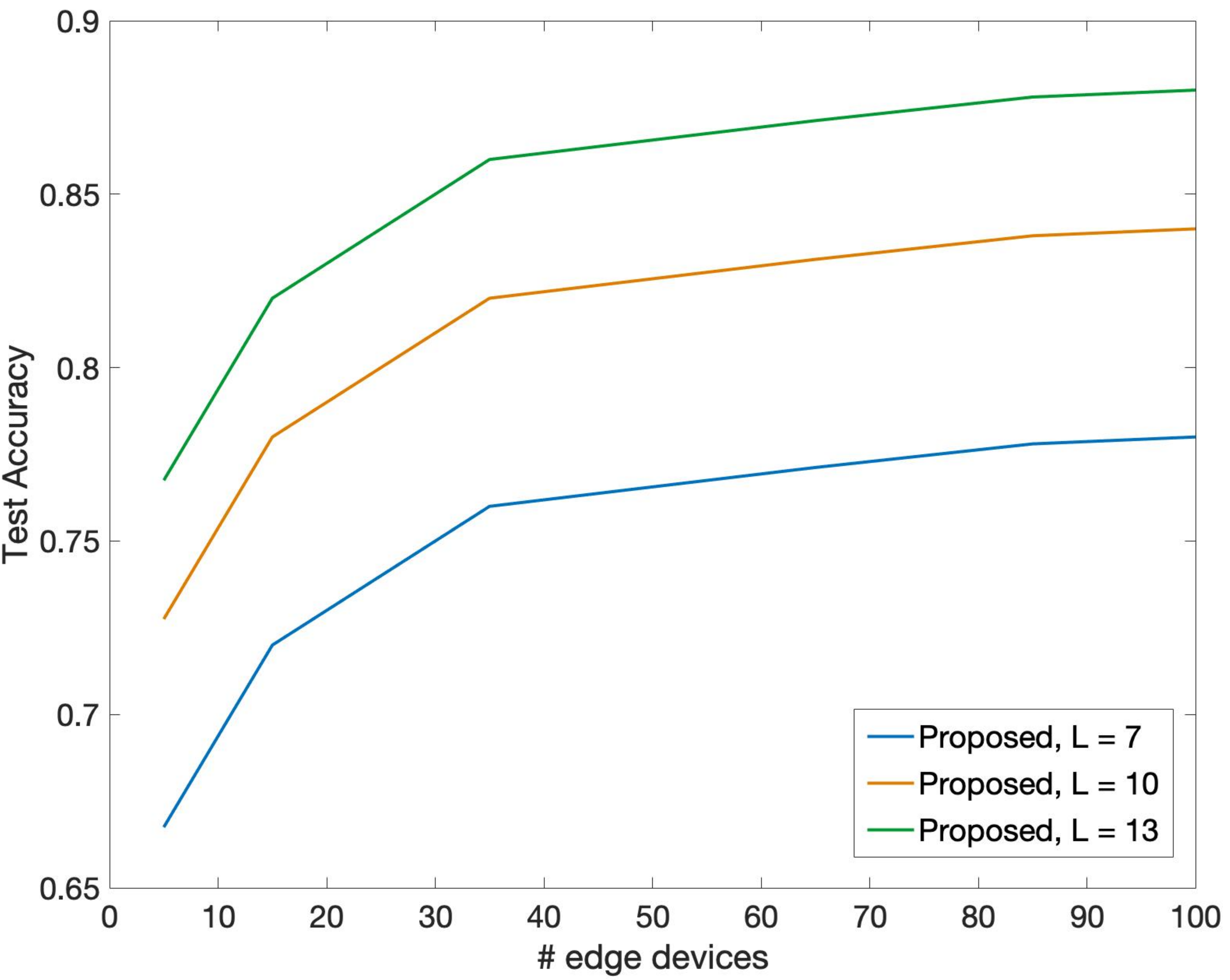} 
     \caption{Effect of the edge-device number $K$ with $B_{\bs} = 5$ bits and $B_{\rho}  = 26$ bits.}
  \vspace{1mm}
  \label{Fig:effect_user_number}
\end{figure}

\section{Concluding Remarks}\label{Conclusion}
In the context of FEEL, by investigating the statistic distribution of
the normalized stochastic gradient, we propose a novel vector
quantization scheme for high-dimensional stochastic gradients. This quantization scheme is of low-complexity, communication-efficient, and convergence-warranted. This work represents the first attempt to quantize high-dimensional stochastic gradients using efficient Grassmannian quantization, which is shown to be more communication-efficient than its state-of-the-art scalar counterpart. In the future, this work can be generalized into applying vector quantization to the accumulated quantization error, which is used for accelerating learning. Moreover, the vector quantization scheme can be further developed by taking the sparsity property and temporal correlation of stochastic gradients into consideration.

\appendix
\subsection{Proof of Lemma~\ref{Lemma: Uniformity}}\label{Proof: Lemma_uniformity}
The normalized stochastic gradient can be written as $\bff = \frac{\bx}{\lVert \bx\rVert}$ due to its uniformity on the Grassmann manifold, where elements of $\bx$ are i.i.d. Gaussian distributed with zero mean and unit variance. Thereby, an arbitrary block gradient can be written as $\bv = \l[\frac{x_m}{\lVert\bx\rVert},\frac{x_{m+1}}{\lVert\bx\rVert},\cdots,\frac{x_n}{\lVert\bx\rVert}\r]^T$ with its norm being $\lVert\bv\rVert = \frac{\sqrt{\sum_{i=m}^nx^2_i}}{\lVert\bx\rVert}$. Then, it follows that 
\begin{equation}
\bs = \frac{\bv}{\lVert\bv\rVert} =  \l[\frac{x_m}{\sqrt{\sum_{i=m}^nx^2_i}},\frac{x_{m+1}}{\sqrt{\sum_{i=m}^nx^2_i}},\cdots,\frac{x_n}{\sqrt{\sum_{i=m}^nx^2_i}}\r]^T.
\end{equation}
This implies that $\bs = \frac{\bv}{\lVert\bv\rVert}$ is uniformly distributed on the Grassmann manifold.

\subsection{Proof of Proposition~\ref{Proposition:hyper-sperical_cap}}\label{Proof:hyper-sperical_cap}
To begin with, by applying the \emph{law of large numbers}, it is easy to show that the hinge vector will converge to the constant vector $\bh_{\sf ref}=\frac{1}{\sqrt{M}}\pmb{1}_{M\times1}$, as $L\to \infty$. Next, we focus on calculating the convergence rate.
According to the definition of the chordal distance, we have
\begin{align}
P_1(r) & = \text{Pr}\l(d_c(\bh_{\sf ref},\bh)>r\r) \nn\\
&= \text{Pr}\l(\sum_{i=1}^M\frac{h_i}{\sqrt{M}}<\sqrt{1-r^2}\r) \nn\\
& =  \text{Pr}\l(\frac{1}{\sqrt{M}}\sum_{i=1}^M\sqrt{\frac{z_i}{\sum_{j=1}^{M}z_j}}<b\r), 
\end{align}
where $b^2 = 1- r^2 = \frac{1-\epsilon_L}{1+\epsilon_L}$ with
$\epsilon_L = \frac{r^2}{2-r^2}$ and $h_i =
\sqrt{\frac{z_i}{\sum_{j=1}^{M}z_j}}$ with $z_i,z_j \sim \mathcal
X^2(L), \forall i,j$. Since $h_i$ and $\sum_{j=1}^{M}z_j$
are independent, we have
\begin{align}
 \text{Pr}\l(d_c(\bh_{\sf
 ref},\bh)>r\r)\cdot\text{Pr}\l(\frac{1}{LM}\sum_{j=1}^M{z_j}<1+\epsilon_L\r)
 =  \text{Pr}\l(d_c(\bh_{\sf ref},\bh)>r, \frac{1}{LM}\sum_{j=1}^M{z_j}<1+\epsilon_L\r). 
\end{align}
Since
\begin{align}
 \text{Pr}\l(d_c(\bh_{\sf ref},\bh) >r,
 \frac{1}{LM}\sum_{j=1}^M{z_j}<1+\epsilon_L\r)
\leq \text{Pr}\l(\frac{1}{\sqrt{M}}\sum_{i=1}^M\sqrt{z_i}<b\sqrt{LM(1+\epsilon_L)}\r), 
 \end{align}
one upper bound can be derived as follows
\begin{equation}\label{P_1_upper}
\text{Pr}\l(d_c(\bh_{\sf ref},\bh)>r\r) \leq \frac{\text{Pr}\l(\frac{1}{\sqrt{M}}\sum_{i=1}^M\sqrt{z_i}<b\sqrt{LM(1+\epsilon_L)}\r)}{\text{Pr}\l(\frac{1}{LM}\sum_{j=1}^M{z_j}<1+\epsilon_L\r)}.
\end{equation}
In the following, we calculate the numerator and the denominator,
respectively. First, we derive an upper bound of the numerator. Define $y_z = \frac{1}{M}\sum_{i=1}^M\sqrt{z_i}$, it follows that
\begin{align}
\text{Pr}\l(\frac{1}{\sqrt{M}}\sum_{i=1}^M\sqrt{z_i}<b\sqrt{LM(1+\epsilon_L)}\r) & = \text{Pr}\l(y_z<\sqrt{L(1-\epsilon_L)}\r)\nn\\
& = \text{Pr}\l(y_z - \mu_{Y} < -\l(\mu_{Y}-\sqrt{L(1-\epsilon_L)}\r)\r)\nn\\
& \leq \text{Pr}\l((y_z - \mu_{Y})^2 > \l(\mu_{Y}-\sqrt{L(1-\epsilon_L)}\r)^2\r),
\end{align}
where the first equality holds given $b^2 =
\frac{1-\epsilon_L}{1+\epsilon_L}$; $\mu_Y = \mathbb E[y_z] = \mathbb E [\frac{1}{M}\sum_{i=1}^M\sqrt{z_i}]$. Then, by \emph{Chebyshev's inequality}, it follows that
\begin{align}\label{FZ_upper}
 \text{Pr}\l(\frac{1}{\sqrt{M}}\sum_{i=1}^M\sqrt{z_i}<b\sqrt{LM(1+\epsilon_L)}\r) \leq \frac{\sigma^2_Y}{ \l(\mu_Y-\sqrt{L(1-\epsilon_L)}\r)^2},
\end{align}
where $\sigma^2_Y = \mathbb E[y^2_z] - \mathbb E^2[y_z]$. In order to derive the closed-form solution for the above bound, it suffices to calculate $\mu_Y$ and $\sigma^2_Y$. Given $z_i \sim \mathcal X^2(L)$, one can have
\begin{equation}
\mu_Y = \mathbb E[\sqrt{z_i}] = \frac{\sqrt{2}\Gamma(\frac{L+1}{2})}{\Gamma(\frac{L}{2})}.
\end{equation}

Then,  by \emph{Stirling's approximation} for gamma function, i.e. $\Gamma(x) \approx \sqrt{2\pi}x^{x-\frac{1}{2}}e^{-x}$, as $x\to \infty$, one can further have that 
\begin{align}\label{mu_lower}
\mu_Y & = \sqrt{L} \cdot e^{\frac{L}{2}\ln (1+\frac{1}{L})-\frac{1}{2}} \overset{(a)}{\geq} \sqrt{L} \cdot e^{-\frac{1}{4L}} \overset{(b)}{\geq} \sqrt{L}\l(1-\frac{1}{4L}\r),
\end{align}
where $(a)$ follows from the fact that $ \ln (1+\frac{1}{L}) \geq \frac{1}{L} - \frac{1}{2L^2}$ and $(b)$ follows from the fact that $e^{-\frac{1}{4L}} \geq 1-\frac{1}{4L}$. Next, we calculate $\sigma^2_Y = \mathbb E[y^2_z]-\mathbb E^2[y_z]$ as follows
\begin{align}
\sigma^2_Y = \mathbb E\l[\frac{1}{M^2}\sum_{i = 1}^{M}z_i + \frac{1}{M^2}\sum_{i\neq j}\sqrt{z_iz_j} \r]-\mathbb E^2[\sqrt{z_i}].
\end{align}
Since $z_i$ and $z_j$ are independent $\mathcal X^2(L)$ distributed random variables, the above equation can be further simplified as
\begin{align}\label{fz_upper}
\sigma^2_Y = \frac{L}{M} - \frac{1}{M}\mu^2_Y \leq \frac{1}{2M} + O\l(\frac{1}{M}L^{-1}\r),
\end{align}
where the inequality follows from~\eqref{mu_lower}.

In the following, we aim to derive a lower bound on $ \l(\mu_Y-\sqrt{L(1-\epsilon_L)}\r)^2$. By~\eqref{mu_lower}, one can have
\begin{equation}
 \l(\mu_Y-\sqrt{L(1-\epsilon_L)}\r)^2 \geq \l(\sqrt{L}-\frac{1}{4\sqrt{L}} -\sqrt{L(1-\epsilon_L)}\r)^2
\end{equation}
Due to the fact that $\sqrt{1-\epsilon_L} \leq1-\frac{\epsilon_L}{2}$ and further write $\epsilon_L = \frac{L^{-\frac{1}{4}}}{2}$, the following result holds.
\begin{equation}\label{fm_lower}
 \l(\mu_Y-\sqrt{L(1-\epsilon_L)}\r)^2 \geq \l(\frac{L^{\frac{1}{4}}}{4}-\frac{1}{4\sqrt{L}}\r)^2 \geq  \frac{1}{2M}L^{\frac{1}{2}},
\end{equation}
where the second inequality follows from the fact that $\frac{L^{\frac{1}{4}}}{4}-\frac{1}{4\sqrt{L}} \geq \frac{L^{\frac{1}{4}}}{10} \geq  \frac{1}{\sqrt{2M}}L^{\frac{1}{4}}$ with $M\geq 50$. It is further noted that the condition, i.e. $M\geq 50$, can always hold in our scenario. Substitute~\eqref{fz_upper} and~\eqref{fm_lower} into~\eqref{FZ_upper}, one can have that
\begin{equation}\label{FM_lower}
\text{Pr}\l(\frac{1}{\sqrt{M}}\sum_{i=1}^M\sqrt{z_i}<b\sqrt{LM(1+\epsilon_L)}\r) \leq L^{-\frac{1}{2}} + O\l(L^{-\frac{3}{2}}\r).
\end{equation}

%

Next, we calculate the denominator. It follows from the fact $\sum_{i=1}^M z_i = \sum_{i=1}^{LM} y_i$ with $y_i \sim \mathcal X^2(1)$ that
\begin{align}
\text{Pr}\l(\frac{1}{LM}\sum_{j=1}^M{z_j}<1+\epsilon_L\r) & = 1 - \text{Pr}\l(\sum_{i=1}^{LM}{y_i}\geq LM(1+\epsilon_L)\r).
\end{align}
Then, by applying the \emph{Chernoff bound}, the following result holds
\begin{align}
\text{Pr}\l(\sum_{i=1}^{LM}{y_i}\geq LM(1+\epsilon_L)\r) \leq e^{ML\cdot\min_{t} \l\{\l[-\frac{1}{2}\ln (1-2t) -t (1+\epsilon_L)\r]\r\}}, \quad t\in[0,\frac{1}{2}),
\end{align}
where the minimum is obtained by setting $t = \frac{1}{2}(1-\frac{1}{1+ \epsilon_L})$. Substituting into the above inequality, one can have that 
\begin{equation}
\text{Pr}\l(\sum_{i=1}^{LM}{y_i}\geq LM(1+\epsilon_L)\r) \leq e^{ML[\frac{1}{2}\ln (1+ \epsilon_L)-\frac{1+\epsilon_L}{2} +\frac{1}{2}]}.
\end{equation}
Due to the fact that $\frac{1}{2}\ln (1+ \epsilon_L)-\frac{1+\epsilon_L}{2} +\frac{1}{2} \leq \frac{\ln 2 -1}{2}\epsilon^2_L, \forall \epsilon_L\in [0,1]$, one can further have
\begin{equation}
\text{Pr}\l(\sum_{i=1}^{LM}{y_i}\geq LM(1+\epsilon_L)\r) \leq e^{-\frac{(1-\ln 2)M }{2}\epsilon^2_LL}.
\end{equation}
 Thereby, the denominator can be bounded as 
 \begin{align}\label{LowerBound_De}
\text{Pr}\l(\frac{1}{LM}\sum_{j=1}^M{z_j}<1+\epsilon_L\r) \geq1 - e^{-\frac{(1-\ln 2)M }{2}\epsilon^2_LL}.
\end{align}

Since $\epsilon_L = \frac{L^{-\frac{1}{4}}}{2}$ and substitute~\eqref{LowerBound_De} and~\eqref{FM_lower} into~\eqref{P_1_upper}, $P_1(r)$ can be bounded as 
\begin{equation}
P_1(r) \leq \frac{L^{-\frac{1}{2}} + O(L^{-\frac{3}{2}})}{1+o(1)} \approx L^{-\frac{1}{2}} + O(L^{-\frac{3}{2}}).
\end{equation}
This completes the whole proof.
\subsection{Proof of Lemma~\ref{Q-error_Hinge}}\label{Proof: Q-error_Hinge}
Directly calculate~\eqref{Union_Bound} is difficult. Instead, we aim to bound the four 
terms in~\eqref{Union_Bound}, respectively, in the following. 

\subsubsection{Calculation of $P_1(r)$}

By Proposition~\ref{Proposition:hyper-sperical_cap}, one can have that $P_1(r) \leq  L^{-\frac{1}{2}}+O(L^{-\frac{3}{2}})$.

\subsubsection{Calculation of $P_2(\alpha,r)$} define $\zeta = \text{Pr}(d_c(\bc,\href)\leq (1+\alpha)r)$, it follows that
\begin{align}
P_2(\alpha,r) &= (1-\zeta)^{N'} = e^{N'\ln(1-\zeta)} \leq e^{-N'\zeta}.\label{Upper_P2}
\end{align}

Moreover, since the codewords are uniformly distributed on the Grassmann manifold, one can calculate $\zeta$ as follows
\begin{align}
\zeta  = \frac{A_{\Bhref( (1+\alpha)r)}}{A_{\text{Manifold}}}= \frac{[(1+\alpha)r]^{\frac{M-1}{2}}}{2\sqrt{\pi}(\frac{M-1}{2})^{\frac{1}{2}}},
\end{align}
where $A_{\Bhref( (1+\alpha)r)}$ and $A_{\text{Manifold}}$ denote the 
surface areas of the ball $\Bhref( (1+\alpha)r)$ and the Grassmann manifold, respectively. 
Furthermore, let $1 + \alpha = r^{-\frac{1}{2}}$ with $r  \approx L^{-\frac{1}{8}}$, 
one can have that
\begin{align}
\zeta & = \frac{[(1+\alpha)r]^{\frac{M-1}{2}}}{2\sqrt{\pi}(\frac{M-1}{2})^{\frac{1}{2}}}= \frac{1}{2\sqrt{\pi}}\l(\frac{1}{L}\r)^{\frac{M-1}{32}}\l(\frac{M-1}{2}\r)^{-\frac{1}{2}}.
\end{align}

Next, set $N' = 2^{B_\bh+1}L^{\frac{M-1}{4}+1} (\frac{M-1}{2})^{\frac{1}{2}}\sqrt{\pi}$, it follows from~\eqref{Upper_P2} that
\begin{align}
P_2(\alpha,r) \leq  e^{-N'\zeta} = e^{-2^{B_\bh}L^{1+\frac{7(M-1)}{32}}} \leq e^{-L}.
\end{align}

\subsubsection{Calculation of $P_3(\alpha,r)$} define $X$ the random variable, which indicates the total number of codewords
lying inside $\Bhref( (1+\alpha)r)$. Thereby, it follows that $X\sim Binomial(N',\zeta)$. Given large $N'$, $X$
can be further approximated as a gaussian distributed random variable, i.e. $X\sim {\mathcal N}(N'\zeta,N'\zeta(1-\zeta))$. Then,
one can have that
\begin{align}
P_3(\alpha,r) & = \text{Pr}(X\geq N)\nn\\
& = \text{Pr}\l(\frac{X-N'\zeta}{\sqrt{N'\zeta(1-\zeta)}}\geq \frac{N-N'\zeta}{\sqrt{N'\zeta(1-\zeta)}}\r)\nn\\
& \leq e^{-\frac{(N-N'\zeta)^2}{2N'\zeta(1-\zeta)}}\nn\\
& \leq e^{-(\frac{N'\zeta}{2}-N)}.\label{P_3_Upper}
\end{align}

Recall that $N' = 2^{B_\bh+1}L^{\frac{M-1}{4}+1} (\frac{M-1}{2})^{\frac{1}{2}}\sqrt{\pi}$ and $\zeta = \frac{1}{2\sqrt{\pi}}\l(\frac{1}{L}\r)^{\frac{M-1}{32}}\l(\frac{M-1}{2}\r)^{-\frac{1}{2}}$, one can have that
\begin{align}
\frac{N'\zeta}{2}-N  > 2^{B_\bh}\l(\frac{L}{2}-1\r) \geq \frac{L}{2}-1.
\end{align}
Substitute the above inequality into~\eqref{P_3_Upper}, the following result holds
\begin{equation}
P_3(\alpha,r) \leq e^{-(\frac{L}{2}-1)}.
\end{equation}

\subsubsection{Calculation of $P_4(\alpha,r) = (1-P_1(r))\left(
  1+\frac{2}{\alpha} \right) D(N')$} since $N'$ codewords are uniformly distributed on the Grassmann manifold, one can have
\begin{align}
D(N') & \leq e^{-\frac{2}{M-1}\ln N'}\nn\\
& = e^{-\frac{2}{M-1}\ln \l(2^{B_\bh+1}L^{\frac{M-1}{4}+1} (\frac{M-1}{2})^{\frac{1}{2}}\sqrt{\pi}\r)}\nn\\
& = e^{-\frac{2}{M-1} \l( \ln 2^{B_\bh} + \ln L^{\frac{M-1}{4}+1}+\ln (\frac{M-1}{2})^{\frac{1}{2}}  + \ln 2\sqrt{\pi}  \r)}\nn\\
& \leq 2^{-\frac{2B_\bh}{M-1}}L^{-\frac{1}{2}-\frac{2}{M-1}}\nn\\
& \leq 2^{-\frac{2B_\bh}{M-1}}L^{-\frac{1}{2}},
\end{align}
where the second inequality follows from the facts that $e^{-\frac{1}{M-1}\ln (\frac{M-1}{2})} \leq 1$ and $e^{-\frac{2}{M-1}\ln 2\sqrt{\pi}} \leq 1$.
Then, substitute $\alpha = r^{-\frac{1}{2}} -1 $ into $P_4(\alpha,r)$, one can have
\begin{equation}
P_4(\alpha,r) \leq (1+\frac{2}{\alpha})D(N') \leq \beta_L2^{-\frac{2B_\bh}{M-1}}L^{-\frac{1}{2}},
\end{equation}
where $\beta_L = \frac{1+r^{\frac{1}{2}}}{1-r^{\frac{1}{2}}}$ with $r =  L^{-\frac{1}{8}}$. Taking the summation on the derived 
four upper bounds, ~\eqref{Upper_hinge} is straightforward.

\subsection{Proof of Theorem~\ref{Theorem: Lower_Upper_Bound}}\label{Proof: Lower_Upper_Bound}
Following from the assumption that the elements of stochastic gradient $\bg$ are i.i.d Gaussian distributed with $0$ mean and unit variance as aforementioned, one can have
\begin{equation}
B \geq \frac{ML}{2}\log_2\frac{1}{D},
\end{equation}
where the inequality follows from the \emph{converse theorem}; $B$ denotes the total number of bits; $D = \frac{\mathbb{E}\left[\lVert\bg-\widehat{\bg}\rVert^2\right]}{ML} = \mathbb E[(g_i-\hat{g}_i)^2], \forall i$. Rearranging the terms, it follows that
\begin{equation}
\ln \frac{\mathbb{E}\left[\lVert\bg-\widehat{\bg}\rVert^2\right]}{ML}  \geq -\frac{2\ln 2}{ML}B.
\end{equation}

Next, we focus on calculation of the upper bound. To begin with, one can show that 
\begin{align}
 \frac{\frac{\rho^2_{\max}}{4}2^{-2B^*_\rho}}{ML\cdot2^{-\frac{2(B^*_{\bs}-1)}{L-1}}}\leq 1,
\end{align}
where $B^*_\rho$ and $B^*_{\bs}$ is defined in~\eqref{Bit: norm} and~\eqref{Bit: block}, respectively. 
Thereby, one upper bound of $\ln \frac{\mathbb{E}\left[\lVert\bg-\widehat{\bg}\rVert^2\right]}{ML}$ is derived in~\eqref{upper_mse} at the top of next page,
\begin{figure*}
\begin{align}
\ln \frac{\mathbb{E}\left[\lVert\bg-\widehat{\bg}\rVert^2\right]}{ML} 
&\leq  \ln \l(2\cdot2^{-\frac{2( B^*_\bs-1)}{L-1}}  + \beta_L2^{-\frac{2B^*_\bh}{M-1}}L^{-\frac{1}{2}}+ L^{-\frac{1}{2}} + O(L^{-\frac{3}{2}})\r)\nn\\
& = \ln \l(2\cdot2^{-\frac{2( B^*_\bs-1)}{L-1}} \l[1+ \frac{\beta_L}{2}L^{-\frac{1}{2}}2^{\frac{2( B^*_\bs-1)}{L-1}-\frac{2B^*_\hv}{M-1}}  + \frac{1}{2}L^{-\frac{1}{2}}2^{\frac{2( B^*_\bs-1)}{L-1}}+ O(L^{-\frac{3}{2}})\r]\r)\nn\\
& \leq -\frac{2\ln 2}{L-1}B^*_{\bs} + \frac{L+1}{L-1}\ln 2 \nn\\
& + \ln \l(1+ \frac{\beta_L}{2}L^{-\frac{1}{2}}2^{\frac{2( B^*_\bs-1)}{L-1}-\frac{2B^*_\hv}{M-1}}  + \frac{1}{2}L^{-\frac{1}{2}}2^{\frac{2( B^*_\bs-1)}{L-1}}+ O(L^{-\frac{3}{2}})\r)\label{upper_mse}
\end{align}
\hrule
\end{figure*}
where $\beta_L = \frac{1+ L^{-\frac{1}{16}}}{1-L^{-\frac{1}{16}}}$. Furthermore, since we consider the quantization problem in the low resolution regime, i.e. $B\leq ML$, implying $B_s < L$, one can have that $2^{\frac{2( B^*_\bs-1)}{L-1}}<4$. Then, it follows that one upper bound of the third term in~\eqref{upper_mse} can be derived in~\eqref{Upper_third_term}.
\begin{figure*}
\begin{align}\label{Upper_third_term}
& \ln \l(1+ \frac{\beta_L}{2}L^{-\frac{1}{2}}2^{\frac{2( B^*_\bs-1)}{L-1}-\frac{2B^*_\hv}{M-1}}  + \frac{1}{2}L^{-\frac{1}{2}}2^{\frac{2( B^*_\bs-1)}{L-1}}+ O(L^{-\frac{3}{2}})\r)\nn\\
 & \leq \ln \l(1+ 2\beta_LL^{-\frac{1}{2}} + 2L^{-\frac{1}{2}}  + O(L^{-\frac{3}{2}})\r) \nn\\
&  \leq  2(\beta_L+1)L^{-\frac{1}{2}}-2(\beta_L+1)^2L^{-1} + O(L^{-\frac{3}{2}}).
\end{align}
\hrule
\end{figure*}
Substituting~\eqref{Upper_third_term} into~\eqref{upper_mse}
with~\eqref{Bit: block} involved, we have
\begin{equation}
\ln \frac{\mathbb{E}\left[\lVert\bg-\widehat{\bg}\rVert^2\right]}{ML}  \leq c_{\sf gap} - \frac{2 \ln 2}{ML}B +O\l(L^{-\frac{3}{2}}\r),\quad L\to \infty,
\end{equation}
where $c_{\sf gap} =  \ln 2-\ln \frac{2L}{L-1} +\frac{2}{ML}\ln \frac{ML+\sqrt{ML}}{2}+\frac{L-1}{L}\ln \frac{2L}{L-1}+\frac{2}{L}\ln 2+\frac{2(M-1)}{ML}\ln \frac{2}{M-1}+\frac{\ln 2}{ML}+ \frac{M-1}{ML}\ln \beta_LM\sqrt{L} + 2(\beta_L+1)L^{-\frac{1}{2}} -2(\beta_L+1)^2L^{-1}$ is a constant given $M$ and $L$. This completes the whole proof.
\subsection{Proof of Theorem~\ref{Convergence}}\label{Proof: Convergence}
Take Assumption~\ref{Assump:smoothness}, one can have that 
\begin{equation}
F_{n+1} - F_{n} \leq \bar{\bg}^T_n(\boldsymbol\theta_{n+1}-\boldsymbol\theta_{n}) + \sum_{i=1}^{\sf Dim}\frac{l_i}{2}(\boldsymbol\theta_{n+1}-\boldsymbol\theta_{n})^2_i,
\end{equation}
where $F_{n}$ denotes the global objective at the $n$-th iteration; $\bar{\bg}_n = \nabla F_{n}$ is the gradient of the global objective and $\boldsymbol\theta_{n} = \boldsymbol\theta [n]$ is the model parameter; $(\boldsymbol\theta_{n+1}-\boldsymbol\theta_{n})^2_i$ denotes the square of the $i$-th coefficient of $\boldsymbol\theta_{n+1}-\boldsymbol\theta_{n}$.

By taking the expectation under the current model $\boldsymbol\theta_{n}$ on the both side of the above inequality, it follows that 
\begin{align}\label{Eq:Upper_whole_term}
\mathbb E \l[F_{n+1} - F_{n}|\boldsymbol\theta_{n}\r] &\leq \mathbb E \l[\bar{\bg}^T_n(\boldsymbol\theta_{n+1}-\boldsymbol\theta_{n})|\boldsymbol\theta_{n}\r]\nn\\
 &+ \mathbb E \l[\sum_{i=1}^{\sf Dim}\frac{l_i}{2}(\boldsymbol\theta_{n+1}-\boldsymbol\theta_{n})^2_i|\boldsymbol\theta_{n}\r].
\end{align}

In the following, we treat $\mathbb E \l[\bar{\bg}^T_n(\boldsymbol\theta_{n+1}-\boldsymbol\theta_{n})|\boldsymbol\theta_{n}\r]$ and $\mathbb E \l[\sum_{i=1}^{\sf Dim}\frac{l_i}{2}(\boldsymbol\theta_{n+1}-\boldsymbol\theta_{n})^2_i|\boldsymbol\theta_{n}\r]$, separately.

\underline{\textbf{(1) Calculation of} $\mathbb E \l[\bar{\bg}^T_n(\boldsymbol\theta_{n+1}-\boldsymbol\theta_{n})|\boldsymbol\theta_{n}\r]$: }

Plug in the model-update step~\eqref{ModelUpdate}, one can have that
\begin{equation}
\mathbb E \l[\bar{\bg}^T_n(\boldsymbol\theta_{n+1}-\boldsymbol\theta_{n})|\boldsymbol\theta_{n}\r] = \mathbb E \l[\frac{-\eta}{K}\bar{\bg}^T_n\sum_{k = 1}^{K}\widehat{\bg}^{(k)}_n|\boldsymbol\theta_{n}\r],
\end{equation}
where $\widehat{\bg}^{(k)}_n = \l\lVert\widehat{\bg}^{(k)}_n\r\rVert\widehat{\bff}^{(k)}_n$ denotes the quantized version of the stochastic gradient from the $k$-th edge device at the $n$-th iteration. Let $\widehat{\bg}^{(k)}_n = {\bg}^{(k)}_n - \boldsymbol\Delta^{(k)}_{n}$ with $\lVert\boldsymbol\Delta^{(k)}_{n}\rVert^2$ denoting the quantization error, the above equation can be rewritten as 
\begin{equation}\label{Eq:inter}
\mathbb E \l[\bar{\bg}^T_n(\boldsymbol\theta_{n+1}-\boldsymbol\theta_{n})|\boldsymbol\theta_{n}\r] = \mathbb E \l[\frac{-\eta}{K}\bar{\bg}^T_n\sum_{k = 1}^{K}\l({\bg}^{(k)}_n - \boldsymbol\Delta^{(k)}_{n}\r)|\boldsymbol\theta_{n}\r].
\end{equation}

Furthermore, given $\lVert\widehat{\bg}^{(k)}_n\rVert = \lVert{\bg}^{(k)}_n\rVert - \Delta\rho_{n}^{(k)}$ and $\widehat{\bff}^{(k)}_n = {\bff}^{(k)}_n - \Delta\bff_{n}^{(k)}$ with $\mathbb{E}[\Delta\rho^{(k)}_n] = 0$ and $\mathbb{E}[\Delta\bff^{(k)}_n] = {\bf 0} \in \mathbb{R}^{{\sf Dim}\times 1}, \forall n,k$, it follows that 
\begin{equation}
\mathbb E [\widehat{\bg}^{(k)}_n] = \mathbb E \l[\l( \lVert{\bg}^{(k)}_n\rVert - \Delta\rho_{n}^{(k)}\r)\l( {\bff}^{(k)}_n - \Delta\bff_{n}^{(k)}\r)\r] = \mathbb E [{\bg}^{(k)}_n].
\end{equation}
Due to the fact that $\mathbb E [\widehat{\bg}^{(k)}_n] = \mathbb E [{\bg}^{(k)}_n] - \mathbb E [\boldsymbol\Delta^{(k)}_n]$, one can conclude that $\mathbb E [\boldsymbol\Delta^{(k)}_n] = {\bf 0}  \in \mathbb{R}^{{\sf Dim}\times 1}$. Thereby, it follows from~\eqref{Eq:inter} that 
\begin{equation}\label{Eq: first_term}
\mathbb E \l[\bar{\bg}^T_n(\boldsymbol\theta_{n+1}-\boldsymbol\theta_{n})|\boldsymbol\theta_{n}\r] = \mathbb E \l[\frac{-\eta}{K}\bar{\bg}^T_n\sum_{k = 1}^{K}{\bg}^{(k)}_n|\boldsymbol\theta_{n}\r]= -\eta\lVert \bar{\bg}_n\rVert^2,
\end{equation}
where the second equality follows from Assumption~\ref{Assump:variance_bound} that $\mathbb E[\bg^{(k)}_n] = \bar{\bg}_n, \forall k$.

\underline{\textbf{(2) Calculation of} $\mathbb E \l[\sum_{i=1}^{\sf Dim}\frac{l_i}{2}(\boldsymbol\theta_{n+1}-\boldsymbol\theta_{n})^2_i|\boldsymbol\theta_{n}\r]$:}

Let $l_0 = \lVert\bl\rVert_{\infty}$ with $\bl$ defined in Assumption~\ref{Assump:variance_bound}, one can have that
\begin{equation}
\mathbb E \l[\sum_{i=1}^{\sf Dim}\frac{l_i}{2}(\boldsymbol\theta_{n+1}-\boldsymbol\theta_{n})^2_i|\boldsymbol\theta_{n}\r]\leq \mathbb E \l[ \frac{l_0}{2} \lVert\boldsymbol\theta_{n+1}-\boldsymbol\theta_{n}\rVert^2|\boldsymbol\theta_{n}\r]
\end{equation}

Since $\boldsymbol\theta_{n+1}-\boldsymbol\theta_{n} = \frac{\eta}{K}\sum_{k = 1}^{K}\widehat{\bg}^{(k)}_n$ with $\widehat{\bg}^{(k)}_n$ denoting the quantized stochastic gradient from the $k$-th local device at the $n$-th iteration, it follows that 
\begin{align}\label{Eq:Converge_Upper_inter}
\mathbb E \l[\sum_{i=1}^{\sf Dim}\frac{l_i}{2}(\boldsymbol\theta_{n+1}-\boldsymbol\theta_{n})^2_i|\boldsymbol\theta_{n}\r]&\leq\mathbb E \l[ \frac{l_0}{2} \lVert\boldsymbol\theta_{n+1}-\boldsymbol\theta_{n}\rVert^2|\boldsymbol\theta_{n}\r] \nn\\
& = \mathbb E \l[ \frac{\eta^2l_0}{2K^2}\l\lVert\sum_{k = 1}^{K}\widehat{\bg}^{(k)}_n\r\rVert^2|\boldsymbol\theta_{n}\r].
\end{align}

Due to the fact that arbitrary two different high-dimensional vectors are quasi-orthogonal, the following result holds
\begin{equation}\label{Eq:quasi_transform}
\l\lVert\sum_{k = 1}^{K}\widehat{\bg}^{(k)}_n\r\rVert^2 = \l(\sum_{k = 1}^{K}\widehat{\bg}^{(k)}_n\r)^T\l(\sum_{k = 1}^{K}\widehat{\bg}^{(k)}_n\r)\approx \sum_{k = 1}^{K}\l\lVert{\widehat\bg}^{(k)}_n\r\rVert^2.
\end{equation}

Substituting~\eqref{Eq:quasi_transform} into~\eqref{Eq:Converge_Upper_inter}, the following inequality follows
\begin{equation}\label{Ineq: second_term}
\mathbb E \l[\sum_{i=1}^{\sf Dim}\frac{l_i}{2}(\boldsymbol\theta_{n+1}-\boldsymbol\theta_{n})^2_i|\boldsymbol\theta_{n}\r]\leq  \frac{\eta^2l_0}{2K^2}\sum_{k = 1}^{K}\mathbb E \l[\l\lVert{\widehat\bg}^{(k)}_n\r\rVert^2|\boldsymbol\theta_{n}\r].
\end{equation}

Moreover, since $\widehat{\bg}^{(k)}_n = {\bg}^{(k)}_n - \boldsymbol\Delta^{(k)}_{n}$, one can have that
\begin{equation}
\lVert\widehat{\bg}^{(k)}_n\rVert^2 = ({\bg}^{(k)}_n - \boldsymbol\Delta^{(k)}_{n})^T({\bg}^{(k)}_n - \boldsymbol\Delta^{(k)}_{n}) \overset{(a)}{\approx} \lVert{\bg}^{(k)}_n\rVert^2 + \lVert\boldsymbol\Delta^{(k)}_n\rVert^2,\nn
\end{equation}
where $(a)$ follows from the same argument that high-dimensional vectors are quasi-orthogonal. Then, the inequality~\eqref{Ineq: second_term} can be simplified as 
\begin{align}
\mathbb E \l[\sum_{i=1}^{\sf Dim}\frac{l_i}{2}(\boldsymbol\theta_{n+1}-\boldsymbol\theta_{n})^2_i|\boldsymbol\theta_{n}\r] & \leq  \frac{\eta^2l_0}{2K^2}\sum_{k = 1}^{K}\mathbb E \l[\l\lVert{\bg}^{(k)}_n\r\rVert^2|\boldsymbol\theta_{n}\r] \nn\\
&+ \frac{\eta^2l_0}{2K^2}\sum_{k = 1}^{K}{\mathbb E}\l[\l\lVert\bg-\widehat\bg\r\rVert^2\r],
\end{align}
where ${\mathbb E}[\lVert\bg-\widehat\bg\rVert^2] = \mathbb E [\lVert\boldsymbol\Delta^{(k)}_n\rVert^2|\boldsymbol\theta_{n}],\forall n,k$. Next, by Assumption~\ref{Assump:variance_bound}, it can be obtained that $\mathbb E [\lVert{\bg}^{(k)}_n\rVert^2|\boldsymbol\theta_{n}] = \lVert \boldsymbol\sigma\rVert^2+\lVert \bar{\bg}_n\rVert^2$ with $\bar{\bg}_n$ denoting the gradient of at the $n$-th iteration. Then, one can have that
\begin{align}\label{Eq:Upper_second_term}
\mathbb E \l[\sum_{i=1}^{\sf Dim}\frac{l_i}{2}(\boldsymbol\theta_{n+1}-\boldsymbol\theta_{n})^2_i|\boldsymbol\theta_{n}\r]& \leq \frac{\eta^2l_0}{2K}\l(\lVert \boldsymbol\sigma\rVert^2+\lVert \bar{\bg}_n\rVert^2\r) \nn\\
& + \frac{\eta^2l_0}{2K}{\mathbb E}\l[\l\lVert\bg-\widehat\bg\r\rVert^2\r].
\end{align}

By substituting~\eqref{Eq: first_term} and~\eqref{Eq:Upper_second_term} into~\eqref{Eq:Upper_whole_term}, the following result holds
\begin{align}
\mathbb E \l[F_{n+1} - F_{n}|\boldsymbol\theta_{n}\r] &\leq -\eta\lVert \bar{\bg}_n\rVert^2 + \frac{\eta^2l_0}{2K}{\mathbb E}\l[\l\lVert\bg-\widehat\bg\r\rVert^2\r] \nn\\
& + \frac{\eta^2l_0}{2K}\l(\lVert \boldsymbol\sigma\rVert^2+\lVert \bar{\bg}_n\rVert^2\r).
\end{align}

Next, extend the expectation over randomness in the trajectory, and perform a telescoping sum over the all the iterations, one lower bound on $F_0-F^*$ is derived in~\eqref{Lower_Sum}.
\begin{figure*}
\begin{align}\label{Lower_Sum}
F_0-F^* & \geq F_0 -\mathbb E[F_N]\nn\\
&= \mathbb E\l[\sum_{n=0}^{N-1}\l(F_n-F_{n+1}\r)\r]\nn\\
& \geq\mathbb E\l[\sum_{n=0}^{N-1} \l(\eta\lVert \bar{\bg}_n\rVert^2 - \frac{\eta^2l_0}{2K}{\mathbb E}\l[\l\lVert\bg-\widehat\bg\r\rVert^2\r] - \frac{\eta^2l_0}{2K}\l(\lVert \boldsymbol\sigma\rVert^2+\lVert \bar{\bg}_n\rVert^2\r)\r)\r]
\end{align}
\hrule
\end{figure*}

Substituting $\eta = \frac{1}{\sqrt{l_0N}}$ into the above inequality and rearrange the terms, it follows that
\begin{align}
\mathbb E\l[\frac{1}{N}\sum_{n=0}^{N-1}\lVert\bar{\bg}_n\rVert^2\r] &
\leq  
\frac{\sqrt{l_0}\l(\frac{1}{2K}{\mathbb E}\l[\l\lVert\bg-\widehat\bg\r\rVert^2\r] +\frac{\lVert\boldsymbol\sigma\rVert^2}{2K}+F_0-F^*\r)}{\sqrt{N}-\frac{\sqrt{l_0}}{2K}}.
\end{align}

This completes the whole proof.

\bibliography{reference}
\bibliographystyle{IEEEtran}

\end{document}

%% file: header.tex
\newtheorem{theorem}{Theorem}

\newtheorem{lemma}{Lemma}
\newtheorem{scheme}{Scheme}   
\newtheorem{corollary}{Corollary}

\newtheorem{proposition}{Proposition}

\newtheorem{assumption}{Assumption}

\newtheorem{remark}{\bf Remark}

\def\qed{$\Box$}

\def\proof{\noindent{\emph{Proof:} }}

\def
\endproof{\hspace*{\fill}~\qed
\par
\endtrivlist\unskip}
\def\endproof{\hspace*{\fill}~\qed\par\endtrivlist\vskip3pt}

\def\phi{\varphi}

\def\l{\left}
\def\r{\right}
\def\({\left(}
\def\){\right)}

\def\bc{{\mathbf{c}}}

\def\bff{{\mathbf{f}}}
\def\bg{{\mathbf{g}}}
\def\bh{{\mathbf{h}}}

\def\bl{{\mathbf{l}}}

\def\bs{{\mathbf{s}}}

\def\bv{{\mathbf{v}}}

\def\bx{{\mathbf{x}}}

\def\b0{{\mathbf{0}}}

\newcommand{\m}[1]{\mathbf{#1}}

\newcommand{\hv}{\m{h}}

\newcommand{\nn}{\nonumber}

%% file: Gradient_Quantization.bbl
\begin{thebibliography}{10}
\providecommand{\url}[1]{#1}
\csname url@samestyle\endcsname
\providecommand{\newblock}{\relax}
\providecommand{\bibinfo}[2]{#2}
\providecommand{\BIBentrySTDinterwordspacing}{\spaceskip=0pt\relax}
\providecommand{\BIBentryALTinterwordstretchfactor}{4}
\providecommand{\BIBentryALTinterwordspacing}{\spaceskip=\fontdimen2\font plus
\BIBentryALTinterwordstretchfactor\fontdimen3\font minus
  \fontdimen4\font\relax}
\providecommand{\BIBforeignlanguage}[2]{{%
\expandafter\ifx\csname l@#1\endcsname\relax
\typeout{** WARNING: IEEEtran.bst: No hyphenation pattern has been}%
\typeout{** loaded for the language `#1'. Using the pattern for}%
\typeout{** the default language instead.}%
\else
\language=\csname l@#1\endcsname
\fi
#2}}
\providecommand{\BIBdecl}{\relax}
\BIBdecl

\bibitem{zhu2018towards}
G.~Zhu, D.~Liu, Y.~Du, C.~You, J.~Zhang, and K.~Huang, ``Towards an intelligent
  edge: Wireless communication meets machine learning,'' 2018. [Online].
  Available: https://arxiv.org/pdf/1809.00343.pdf.

\bibitem{konevcny2016federated}
J.~Kone{\v{c}}n{\`y}, H.~B. McMahan, F.~X. Yu, P.~Richt{\'a}rik, A.~T. Suresh,
  and D.~Bacon, ``Federated learning: Strategies for improving communication
  efficiency,'' 2016. [Online]. https://arxiv.org/pdf/1610.05492.pdf.

\bibitem{tandon2017gradient}
R.~Tandon, Q.~Lei, A.~G. Dimakis, and N.~Karampatziakis, ``Gradient coding:
  Avoiding stragglers in distributed learning,'' in \emph{\emph{Proc. of} Intl.
  Conf. Mach. Learning (ICML)}, pp. 3368--3376, Jul. 2017.

\bibitem{chen2018lag}
T.~Chen, G.~Giannakis, T.~Sun, and W.~Yin, ``{LAG}: Lazily aggregated gradient
  for communication-efficient distributed learning,'' in \emph{\emph{Proc. of}
  Adv. Neural Inf. Proc. Systems (NIPS)}, pp. 5050--5060, Dec. 2018.

\bibitem{zhang2017zipml}
H.~Zhang, J.~Li, K.~Kara, D.~Alistarh, J.~Liu, and C.~Zhang, ``{ZIPML}:
  Training linear models with end-to-end low precision, and a little bit of
  deep learning,'' in \emph{\emph{Proc. of} the 34th Intl. Conf. Mach. Learning
  (ICML)}, pp. 4035--4043, Aug. 2017.

\bibitem{alistarh2017qsgd}
D.~Alistarh, D.~Grubic, J.~Li, R.~Tomioka, and M.~Vojnovic, ``{QSGD}:
  Communication-efficient {SGD} via gradient quantization and encoding,'' in
  \emph{\emph{Proc. of} Adv. Neural Inf. Proc. Systems (NIPS)}, pp. 1709--1720,
  Dec. 2017.

\bibitem{2018signsgd}
J.~Bernstein, Y.-X. Wang, K.~Azizzadenesheli, and A.~Anandkumar, ``sign{SGD}:
  Compressed optimisation for non-convex problems,'' \emph{\emph{in Proc. of}
  Intl. Conf. Mach. Learning (ICML)}, vol.~80, pp. 559--568, Jul. 2018.

\bibitem{wu2018error}
J.~Wu, W.~Huang, J.~Huang, and T.~Zhang, ``Error compensated quantized {SGD}
  and its applications to large-scale distributed optimization,'' in
  \emph{\emph{Proc. of} Intl. Conf. Mach. Learning (ICML)}, pp. 5321--5329,
  Jul. 2018.

\bibitem{zheng2019communication}
S.~Zheng, Z.~Huang, and J.~T. Kwok, ``Communication-efficient distributed
  blockwise momentum {SGD} with error-feedback,'' 2019. [Online].
  https://arxiv.org/pdf/1905.10936.pdf.

\bibitem{kingma2014adam}
D.~P. Kingma and J.~Ba, ``{ADAM}: A method for stochastic optimization,'' 2014.
  [Online]. https://arxiv.org/pdf/1412.6980.pdf

\bibitem{gersho2012vector}
A.~Gersho and R.~M. Gray, \emph{Vector quantization and signal
  compression}.\hskip 1em plus 0.5em minus 0.4em\relax Springer Science \&
  Business Media, 2012, vol. 159.

\bibitem{love2008overview}
D.~J. Love, R.~W. Heath, Jr, V.~K. Lau, D.~Gesbert, B.~D. Rao, and M.~Andrews,
  ``An overview of limited feedback in wireless communication systems,''
  \emph{IEEE J. Sel. Areas Commun.}, vol.~26, no.~8, pp. 1341--1365, 2008.

\bibitem{mukkavilli2003beamforming}
K.~K. Mukkavilli, A.~Sabharwal, E.~Erkip, and B.~Aazhang, ``On beamforming with
  finite rate feedback in multiple-antenna systems,'' \emph{IEEE Trans. Inf.
  Theory}, vol.~49, no.~10, pp. 2562--2579, 2003.

\bibitem{love2003grassmannian}
D.~J. Love, R.~W. Heath, Jr, and T.~Strohmer, ``Grassmannian beamforming for
  multiple-input multiple-output wireless systems,'' \emph{IEEE Trans. Inf.
  Theory}, vol.~49, no.~10, p. 2735, 2003.

\bibitem{raghavan2007systematic}
V.~Raghavan, R.~W. Heath, Jr, and A.~M. Sayeed, ``Systematic codebook designs
  for quantized beamforming in correlated {MIMO} channels,'' \emph{IEEE J. Sel.
  Areas Commun.}, vol.~25, no.~7, pp. 1298--1310, 2007.

\bibitem{kim2011mimo}
T.~Kim, D.~J. Love, and B.~Clerckx, ``{MIMO} systems with limited rate
  differential feedback in slowly varying channels,'' \emph{IEEE Trans.
  Commun.}, vol.~59, no.~4, pp. 1175--1189, 2011.

\bibitem{huang2009limited}
K.~Huang, R.~W. Heath, Jr, and J.~G. Andrews, ``Limited feedback beamforming
  over temporally-correlated channels,'' \emph{IEEE Trans. Sig. Process.},
  vol.~57, no.~5, pp. 1959--1975, 2009.

\bibitem{xia2006design}
P.~Xia and G.~B. Giannakis, ``Design and analysis of transmit-beamforming based
  on limited-rate feedback,'' \emph{IEEE Trans. Sig. Process.}, vol.~54, no.~5,
  pp. 1853--1863, 2006.

\bibitem{nam2012joint}
J.~Nam, J.-Y. Ahn, A.~Adhikary, and G.~Caire, ``Joint spatial division and
  multiplexing: Realizing massive {MIMO} gains with limited channel state
  information,'' in \emph{\emph{Proc. of} Annual Conf. Inf. Sciences and
  Systems (CISS)}, pp. 1--6, Mar. 2012.

\bibitem{choi2013noncoherent}
J.~Choi, Z.~Chance, D.~J. Love, and U.~Madhow, ``Noncoherent trellis coded
  quantization: A practical limited feedback technique for massive {MIMO}
  systems,'' \emph{IEEE Trans. Commun.}, vol.~61, no.~12, pp. 5016--5029, 2013.

\bibitem{sim2016compressed}
M.~S. Sim, J.~Park, C.-B. Chae, and R.~W. Heath, Jr, ``Compressed channel
  feedback for correlated massive {MIMO} systems,'' \emph{J. Commun. and
  Networks}, vol.~18, no.~1, pp. 95--104, 2016.

\bibitem{dhillon2008constructing}
I.~S. Dhillon, J.~R. Heath, T.~Strohmer, and J.~A. Tropp, ``Constructing
  packings in {Grassmannian} manifolds via alternating projection,''
  \emph{Experimental mathematics}, vol.~17, no.~1, pp. 9--35, 2008.

\bibitem{lau2004design}
V.~Lau, Y.~Liu, and T.-A. Chen, ``On the design of {MIMO} block-fading channels
  with feedback-link capacity constraint,'' \emph{IEEE Trans. Commun.},
  vol.~52, no.~1, pp. 62--70, 2004.

\bibitem{conway2013sphere}
J.~H. Conway and N.~J.~A. Sloane, \emph{Sphere packings, lattices and
  groups}.\hskip 1em plus 0.5em minus 0.4em\relax Springer Science \& Business
  Media, vol. 290, 2013.

\bibitem{dai2008quantization}
W.~Dai, Y.~Liu, and B.~Rider, ``Quantization bounds on {Grassmann} manifolds
  and applications to {MIMO} communications,'' \emph{IEEE Trans. Inf. Theory},
  vol.~54, no.~3, pp. 1108--1123, 2008.

\end{thebibliography}
